\documentclass[review]{elsarticle}

\usepackage{lineno,hyperref}
\usepackage{verbatim}
\usepackage{amsmath}
\usepackage{amssymb}
\usepackage{amsthm}
\usepackage{color}

\newtheorem{proposition}{Proposition}
\newtheorem{remark}{Remark}

\newcommand\T{\rule{0pt}{2.1ex}}
\newcommand\B{\rule[-0.7ex]{0pt}{0pt}}

\modulolinenumbers[5]

\journal{Journal of Information Processing and Management }

\begin{document}

\begin{frontmatter}

\title{A Two Phase Investment Game for Competitive Opinion Dynamics in Social Networks}
\tnotetext[mytitlenote]{A previous, preliminary, concise  version  of  this  paper  was presented at The Joint International Workshop on Social Influence Analysis and Mining Actionable Insights from Social Networks (workshop with IJCAI-ECAI), Stockholm, Sweden, 2018 \cite{dhamal2018manipulating}.}

\author[]{Swapnil Dhamal\fnref{INRIA,TSP,LIA}\corref{mycorrespondingauthor}}

\cortext[mycorrespondingauthor]{This is to indicate the corresponding author.}
\ead{swapnil.dhamal@gmail.com}

\author{Walid Ben-Ameur\fnref{TSP}}

\author{Tijani Chahed\fnref{TSP}}

\author{Eitan Altman\fnref{INRIA,LIA}}

\fntext[INRIA]{Institut National de Recherche en Informatique et en Automatique, Sophia Antipolis-M\'editerran\'ee, Valbonne 06902, France}
\fntext[TSP]{T\'el\'ecom SudParis, CNRS, Evry 91011, France}
\fntext[LIA]{Laboratoire Informatique d'Avignon, Avignon 84140, France}

\begin{abstract}
We propose a setting for two-phase opinion dynamics in social networks, where a node's final opinion in the first phase acts as its initial biased opinion in the second phase. In this setting, we study the problem of two camps aiming to maximize adoption of their respective opinions, by strategically investing on nodes in the two phases. A node's initial opinion in the second phase naturally plays a key role in determining the final opinion of that node, and hence also of other nodes in the network due to its influence on them. However, more importantly, this bias also determines the effectiveness of a camp's investment on that node in the second phase. In order to formalize this two-phase investment setting, we propose an extension of Friedkin-Johnsen model, and hence formulate the utility functions of the camps. We arrive at a decision parameter which can be interpreted as two-phase Katz centrality. There is a natural tradeoff while splitting the available budget between the two phases. A lower investment in the first phase results in worse initial biases in the network for the second phase. On the other hand, a higher investment in the first phase spares a lower available budget for the second phase, resulting in an inability to fully harness the influenced biases. We first analyze the non-competitive case where only one camp invests, for which we present a polynomial time algorithm for determining an optimal way to split the camp's budget between the two phases. We then analyze the case of competing camps, where we show the existence of Nash equilibrium and that it can be computed in polynomial time under reasonable assumptions. We conclude our study with simulations on real-world network datasets, in order to quantify the effects of the initial biases and the weightage attributed by nodes to their initial biases, as well as that of a camp deviating from its equilibrium strategy. Our main conclusion is that, if nodes attribute high weightage to their initial biases, it is advantageous to have a high investment in the first phase, so as to effectively influence the biases to be harnessed in the second phase.
\end{abstract}


\end{frontmatter}


\section{Introduction
}
\label{sec:ODSNmultiphase_intro}

Studying opinion dynamics in a society is important
to understand and influence 
elections, viral marketing, propagation of ideas and behaviors, etc. 
Social networks play a prime role in determining the opinions of  constituent nodes, since nodes usually update their opinions based on the opinions of their connections \cite{easley2010networks,acemoglu2011opinion}. 
This fact is exploited by camps,
  intending to influence the opinions of these nodes in their favor.
A camp could be, for instance, seeker of funds or votes for a particular cause.
In this paper, we consider two  camps who aim to maximize the adoption of their respective opinions in a social network. 
We consider a strict competition in the space of real-valued opinions, where 
one camp aims to drive the overall opinion of the network towards being positive while the other camp aims to drive it towards  negative;
we refer to them as good and bad camps respectively.
(The terminologies `good' and `bad' camps are used only for  the ease of interpretation that one camp holds a positive opinion while the other camp holds a negative opinion on the real number line; it does not necessarily imply that one camp is more virtuous than the other.)
We consider a well-accepted quantification of 
the overall opinion of a network:
the average or equivalently, the sum of opinion values of the nodes in the social network \cite{gionis2013opinion,grabisch2017strategic,dhamal2018framework}.
Hence, the good and bad camps simultaneously aim to
 maximize and minimize this sum, respectively.

The average or sum of  opinion values is well suited to several applications. 
For instance,
in a fund collection scenario, the magnitude of the opinion value of a node can be interpreted as the amount of funds it is willing to contribute, and its sign would imply the camp towards which it is willing to contribute. 
Here, the objective of the good camp would be to drive the sum of opinion values of the nodes to be as high as possible, so as to
gather maximum funds 
for its concerned cause.
On the other hand,
the objective of
the bad camp would be to drive this sum to be as low as possible, that is, to
convince the population to contribute for an opposing cause.
In  such scenarios, the opinion value of a node 
would be a real-valued number, and the overall opinion of the network is well indicated by the sum or average of the opinion values of its constituent nodes.
While  fund collection  is a particular scenario where a node's opinion is explicitly expressed in the form of its contribution, the underlying principle, in general, applies to  scenarios that involve nodes having certain belief or information. 


In the literature on opinion dynamics in social networks,
there have been efforts to develop models which could determine how the individuals update their opinions based on the opinions of their connections
\cite{acemoglu2011opinion}.
With such an underlying model of opinion dynamics, we consider that each camp  aims to
 maximize the adoption of its opinion in the social network, while accounting for the presence of the competing camp.
 %
 %
 %
A camp could hence act on achieving its objective by strategically investing on selected influential individuals in a social network who could adopt its opinion. This investment could be in the form of money, free products or discounts, attention, convincing discussions, etc. 
Thus given a budget constraint, the strategy of a camp comprises  how much to invest on each node, in the presence of a competing camp who also  invests strategically.
This results in a game, and we are essentially interested in determining the equilibrium strategies of the two camps, from which neither camp would want to unilaterally deviate. Hence, our focus in this paper will be on determining the Nash equilibrium of this game.

\subsection{Motivation
}
In the popular model by
Friedkin and Johnsen
\cite{friedkin1990social,friedkin1997social} (which we will describe later), every node holds an initial bias in opinion. It could have  formed owing to 
the node's fundamental views, experiences,
information from news and other sources, 
opinion dynamics in the past, etc. 
This initial bias  plays an important role in determining a node's final opinion, and consequently the opinions of its neighbors and hence that of its neighbors' neighbors and so on. If nodes give significant weightage to their  biases, the camps would want to manipulate these biases. 
This could be achieved by campaigning in two phases,
 wherein the opinion at the conclusion of the first phase would act as the initial biased opinion for the second phase. Such campaigning is often used during elections and   marketing, in order to gradually drive the nodes' opinions.

In real-world scenarios,
the initial bias of a node often impacts a camp's effectiveness  on that node.
For instance, 
if the initial bias of a node is  positive, 
 the investment made by the good camp is likely to be more effective on it than that made by the bad camp. 
 The  reasoning is on similar lines as that of
 models in which a node pays more attention to opinions that do not differ too much from its own opinion (such as the bounded confidence model~\cite{krause2000discrete}).
 Since a camp's effectiveness depends on the nodes' biases, its investment in the first phase not only manipulates the biases for getting a head start in the second phase, but also the effectiveness of its investment in the second phase.

Furthermore,
with the possibility of campaigning in two phases, a camp could not only decide which nodes to invest on, but also how to split its available budget between the two phases.

\subsection{Related Work}
\label{sec:ODSNmultiphase_relevant}

\paragraph{Opinion dynamics in social networks}

The topic of opinion dynamics has received significant attention in the social networks community.
Xia, Wang, and Xuan \cite{xia2013opinion} give a multidisciplinary review of the field of opinion dynamics as a combination of the social processes and the analytical and computational tools.
A line of work deals with opinion diffusion in social networks under popular models such as the independent cascade and linear threshold \cite{easley2010networks, kempe2003maximizing, guille2013information}.
Another line of work addresses
 continuous-time diffusion models~\cite{gomez2016influence}, and  specifically point process models for diffusion~\cite{farajtabar2017coevolve}.
Lorenz~\cite{lorenz2007continuous} surveys several modeling frameworks concerning continuous opinion dynamics.
Rodriguez and Song
\cite{gomez2015diffusion}
present several diffusion
models
and address
problems such as
network estimation, influence estimation, and
influence control,
using methods from
machine learning, probabilistic modeling, event
history analysis, graph theory, and
network science.

Acemoglu and Ozdaglar \cite{acemoglu2011opinion} review several fundamental models of opinion dynamics,
some  noteworthy ones being
DeGroot \cite{degroot1974reaching}, Voter \cite{holley1975ergodic},  
Friedkin-Johnsen \cite{friedkin1990social,friedkin1997social}, bounded confidence \cite{krause2000discrete}, etc.
In Friedkin-Johnsen model, each node updates its opinion using a weighted  combination of its initial bias and its neighbors' opinions. 
In this paper, we generalize this model to multiple phases,
 while also incorporating
the camps' investments.

\paragraph{Identifying influential nodes}

Problems related to determining influential nodes 
for maximizing opinion adoption 
in social networks have been extensively studied 
in the literature 
\cite{easley2010networks,guille2013information}.
For instance,
Yildiz, Ozdaglar, and Acemoglu~\cite{yildiz2013binary} study the problem of optimal placement of stubborn nodes (whose opinion values stay unchanged) in the discrete binary opinions setting.
Gionis, Terzi, and Tsaparas~\cite{gionis2013opinion} study the problem of identifying such nodes whose positive opinions 
would maximize the overall positive opinion 
in the  network.
Kempe, Kleinberg, and Tardos~\cite{kempe2003maximizing} propose approximation algorithms for identifying influential nodes under the independent cascade and linear threshold models, which has since been followed by a plethora of increasingly efficient techniques~\cite{guille2013information}.

Lynn and Lee
\cite{lynn2016maximizing}
 study influence maximization in the
context of the Ising model,
by treating individual opinions as spins in an Ising system at dynamic
equilibrium; hence the goal is to maximize the magnetization of an Ising system given a budget
of external magnetic field.
Rossi and Ahmed
\cite{rossi2015role}
study techniques for discovering roles in networks by
proposing a  general formulation of roles of nodes based on the similarity of feature representation (in contrast to only the
graph representation as traditionally considered).
%
%
%
%
Abiteboul, Preda, and Cobena
\cite{abiteboul2003adaptive}
 develop an on-line algorithm for computing the  importance of a node, which adapts dynamically to the changes in  the network.
Grindrod et al.
\cite{grindrod2011communicability}
propose a way of extending classical node centrality measures from the literature on static networks, to be applied to evolving networks.
Gleich and Rossi
\cite{gleich2014dynamical}
 propose a dynamical system that captures changes to the  centrality of
nodes as external interest in those nodes vary, thus resulting in a time-dependent set of centrality scores.

The competitive setting has resulted in a number of  
game theoretic studies
\cite{myers2012clash,tzoumas2012game,etesami2016complexity}.
Bharathi, Kempe, and Salek~\cite{bharathi2007competitive} were among the first to study  opinion adoption in social networks from a game theoretic viewpoint.
Goyal, Heidari, and Kearns~\cite{goyal2014competitive}
present a model for the diffusion
of two competing opinions in a social network, in which nodes
first choose whether to adopt either of the  opinions or none
of them, and then choose which opinion to adopt.
Anagnostopoulos, Ferraioli, and Leonardi~\cite{anagnostopoulos2015competitive}
study this model in detail for some of the well-known dynamics.
Ghaderi and Srikant \cite{ghaderi2014opinion} study how the equilibrium of the game
depends on the network structure, nodes' initial opinions, the location of stubborn nodes and the extent of their stubbornness.

Specific to analytically tractable models such as DeGroot,
there have been studies 
in the competitive setting
 to identify influential nodes 
 and the amounts to be invested on them
\cite{grabisch2017strategic,dubey2006competing,bimpikis2016competitive}.
Dhamal et al.~\cite{dhamal2018framework} study a broader framework with respect to one such  model (Friedkin-Johnsen model), while considering a number of practically motivated settings such as those accounting for diminishing marginal returns on investment, adversarial behavior of the competitor, uncertainty regarding system parameters, and bound on the combined investment by the camps on each node.
Our work extends these studies to two phases, by identifying
 influential nodes in the two phases and how much they should be invested on in each phase.

\paragraph{Multi-phase opinion diffusion}

There have been a few studies on adaptive selection of influential nodes for opinion diffusion in multiple phases.
Singer \cite{singer2016influence} presents a survey of such adaptive methodologies. 
Golovin and Krause
\cite{golovin2011adaptive}
introduce adaptive submodularity, which facilitates  adaptive greedy algorithm to provide a performance guarantee.
Seeman and Singer
\cite{seeman2013adaptive} were among the first to dedicatedly study the framework of adaptive node selection.
Rubinstein, Seeman, and Singer
\cite{rubinstein2015approximability}
present 
adaptive  algorithms for selecting
nodes with heterogeneous  costs.
Horel and Singer 
\cite{horel2015scalable} develop scalable methods  
for models in which the influence of a set can be expressed as the sum of  influence of its members.
Correa et al. 
\cite{correa2015adaptive}
show that the adaptivity benefit   is bounded
if every pair of nodes randomly meet at the same rate. 
Badanidiyuru et al.
\cite{badanidiyuru2016locally} 
propose an 
algorithm  based on locally-adaptive policies.

Dhamal, Prabuchandran, and Narahari 
\cite{dhamal2016information} empirically study the problem of 
optimally splitting the available budget between two phases under the popular independent cascade model,
which has been extended to more than two phases in \cite{dhamal2018effectiveness}.
Tong et al.~\cite{tong2016adaptive}
 study adaptive node selection in a dynamic independent cascade model.
Yuan and Tang
\cite{tang2016no}
present 
a framework where  nodes can be selected before termination of an ongoing diffusion.
Sun et al.
\cite{sun2018multi} study the problem of
multi-round influence maximization, where the
goal is
to select nodes for each round to maximize the expected number
of nodes that are influenced in at least one round.
Mondal, Dhamal, and Narahari 
\cite{mondal2017two}
study a
setting
where the first phase is regular diffusion, while the second phase is boosted  using referral incentives.

While the  reasoning behind using multiple phases in these studies is adaptation of node selection strategy based on previous observations, we aim to use multiple phases for manipulating the initial biases of  nodes. This requires a very different conceptual and analytical treatment from the ones in the literature.

To the best of our knowledge, there has not been an analytical study on a rich model such as Friedkin-Johnsen, for opinion dynamics in two phases (not even for single camp).
The most relevant to this study is our earlier work  \cite{dhamal2018optmulti} where, however, a camp's influence on a node is assumed to be independent of the node's bias.
In this paper, we consider a more realistic setting by relaxing this assumption.
An interesting outcome of  relaxing this assumption is that, while the camps' optimal strategies turn out to be mutually independent in \cite{dhamal2018optmulti}, these strategies get coupled in our setting. In other words, the setting in \cite{dhamal2018optmulti} results in a {\em competition\/}, while the one in this paper results in a {\em game\/}.

\subsection{Our Contributions
}

Following are the specific contributions of this paper:

\begin{itemize} 

\item
We formulate the two-phase objective function under Friedkin-Johnsen model, 
where a node's final opinion in the first phase acts as its initial bias for the second phase, and 
the effectiveness of a camp's investment on the node depends on this initial bias.
(Section~\ref{sec:ODSNmultiphase_prob})

\item
For the non-competitive case, we develop a polynomial time algorithm for determining an optimal way to split a camp's budget between the two phases and the nodes to be invested on in the two phases.
(Section~\ref{sec:dep_onecamp})

\item 
For the competitive case involving both the camps, we show the existence of Nash equilibrium, and that it can be computed in polynomial time  under reasonable assumptions. (Section~\ref{sec:dep_2camps})

\item
Using simulations, we illustrate our analytically derived results on real-world  network datasets, and quantify the effects of the initial biases and the weightage attributed by nodes to their initial biases, as well as that of a camp deviating from its equilibrium strategy.
(Section~\ref{sec:ODSNmultiphase_sim})

\end{itemize}

\section{Our Model}
\label{sec:ODSNmultiphase_prob}

Given a social network,
let $V$ be the set of nodes, $n$ be the number of nodes (cardinality of $V$), and $E$ be the set of weighted  directed edges.
Our model can be viewed as a multiphase extension of \cite{dhamal2018framework}, and more broadly, an extension of Friedkin-Johnsen model \cite{friedkin1990social,friedkin1997social}.

\vspace{-2mm}
\subsection{Friedkin-Johnsen Model}
\label{sec:ODSN_model}

In Friedkin-Johnsen model,
prior to the process of opinion dynamics, each node holds a certain bias in its opinion.
We denote this opinion bias of a node $i$ by $z_i^0$, and the weightage that the node attributes to it by $w_{ii}^0$.

The network effect is captured by how much a node is influenced by each of its friends or connections, that is, how much weightage is attributed by a node to the opinion of each of its connections. Let $z_j$ be the opinion held by node $j$, and $w_{ij}$ be the weightage attributed by node $i$ to the opinion of node $j$. The influence on node $i$ owing to node $j$ is given by $w_{ij}z_j$, thus the net influence on $i$ owing to all of its connections is $\sum_{j \in V} w_{ij}z_j$ (where $w_{ij} \neq 0$ only if $j$ is a connection of $i$). 
The directed nature of the edges accounts for the fact that the weightage attributed by node $i$ to the opinion of node $j$ would, in general, be different from the weightage attributed by node $j$ to the opinion of node $i$.
Furthermore, undirected edges are a special case of directed edges, where $w_{ij}=w_{ji}$.
Note that  the edge weights could be negative as well (as justified in \cite{altafini2013consensus,proskurnikov2016opinion}). A negative edge weight $w_{ij}$ can be interpreted as distrust that node $i$ holds on node $j$, that is, $i$ would be driven towards adopting an opinion that is opposite to that held or suggested by $j$. 

Since in Friedkin-Johnsen model,
each node updates its opinion using a weighted convex combination of its bias and its neighbors' opinions,
 the update rule   is given by

\vspace{-3mm}
\begin{small}
\begin{align*}
\forall i \in V :\;
z_i
\leftarrow
w_{ii}^0 z_i^0
+ \sum_{j \in V} w_{ij}z_j
\end{align*}
\end{small}
\vspace{-3mm}

\noindent
We can also write this update rule as a recursion 
 (with iterating integer $\tau \geq  0$):
 
 \vspace{-3mm}
 \begin{small}
 \begin{align*}
 \forall i \in V :\;
 z_{i \, \langle \tau \rangle}
 =
 w_{ii}^0 z_i^0
 + \sum_{j \in V} w_{ij}z_{j \, \langle \tau-1 \rangle} \text{ , where} z_{i \, \langle 0 \rangle} = z_i^0
 \end{align*}
 \end{small}
 \vspace{-3mm}

Since the model follows an opinion update rule,
convergence is often a desirable property.
A standard assumption for guaranteeing convergence is
$
\sum_{j \in V} |w_{ij}| < 1
$.
We will later see how we use this condition in our analysis.

\subsection{Our Extended Model}
\label{sec:ODSNmultiphase_model}

As our opinion dynamics runs in two phases, most parameters have two values, one for each phase. For such a parameter, we denote its value corresponding to phase $p$ using superscript $(p)$, where $p=1$ for the first phase and $p=2$ for the second phase.
For instance, we denote the parameter corresponding to the opinion value of node $i$ after phase $p$ by $z_i^{(p)}$.
Table \ref{tab:notation} presents the notation.

\paragraph{Initial Bias}
Consistent with Friedkin-Johnsen model, every node $i$ holds an initial bias in opinion {$z_i^0 \in \mathbb{R}$} prior to the opinion dynamics process, and attributes a weightage of $w_{ii}^0$ to it. In the two-phase setting, $z_i^0$ acts as the initial bias for the first phase.
We denote the opinion value of node $i$ at the conclusion of the first phase  by $z_i^{(1)}$.
This acts as its initial bias for the second phase.
So by convention, we have $z_i^{(0)}=z_i^0$.
In the context of the fund collection example, 
$z_i^0$ is the initial bias of node $i$ based on its perception of the causes for which the funds are being collected by the camps.
For instance, a very positive $z_i^0$ would mean that, prior to campaigning and opinion dynamics process, node $i$ is  aligned towards contributing significantly to the funds pertaining to the good camp's  cause.

\paragraph{Network Effect}
Consistent with Friedkin-Johnsen model, we consider that the influence on node $i$ owing to node $j$ in phase $p$, is $w_{ij}z_j^{(p)}$,
where $w_{ij}$ is the weightage attributed by node~$i$ to the opinion of its connection $j$. 
In the fund collection example, 
the influence on node $i$ owing to its connection $j$ is the product of node $j$'s opinion regarding how much and which cause to contribute to, and the weightage attributed by node $i$ to node $j$'s opinion based on the level of trust and confidence that node $i$ has on node $j$'s opinion.
The net influence on node $i$ owing to all of its connections is, hence, $\sum_{j \in V} w_{ij}z_j^{(p)}$.

\begin{table}[t]
\caption{Notation table}
\label{tab:notation}
\vspace{-3mm}
\begin{center}
\begin{tabular}{|c|l|}
\hline
\T \B
$z_i^0$ & initial biased opinion of node $i$
\\ \hline
\T \B
$w_{ii}^0$ & weightage attributed by node $i$ to its initial bias
\\ \hline
\T \B
$w_{ij}$ & weightage attributed by node $i$ to the opinion of node $j$
\\ \hline
\T \B
$\theta_i$ & total weightage attributed by node $i$ to the camps' opinions
\\ \hline
\T \B
$w_{ig}^{(p)}$ & weightage attributed by node $i$ to the good camp in phase $p$
\\ \hline
\T \B
$w_{ib}^{(p)}$ & weightage attributed by node $i$ to the bad camp in phase $p$
\\ \hline
\T \B
$x_i^{(p)}$ & investment made by the good camp to directly influence node $i$ in phase $p$
\\ \hline
\T \B
$y_i^{(p)}$ & investment made by the bad camp to directly influence node $i$ in phase $p$
\\ \hline 
\T \B
$k_g$ & total budget of the good camp
\\ \hline 
\T \B
$k_b$ & total budget of the bad camp
\\ \hline 
\T \B
$z_i^{(p)}$ &  opinion of node $i$ at the end of phase $p$
\\ \hline
\end{tabular}
\end{center}
\vspace{-2mm}
\end{table}

\paragraph{Weightage to Campaigning}
We denote the weightage that node $i$ attributes to the good and bad campaigning in phase~$p$ by $w_{ig}^{(p)}$ and $w_{ib}^{(p)}$, respectively.
        Since we consider that the initial bias of a node  impacts the effectiveness of camps' investments, $z_i^{(p-1)}>0$ would likely result in $w_{ig}^{(p)} > w_{ib}^{(p)}$.
Note that $w_{ii}^0$ would also play a role since it quantifies the weightage given by node $i$ to its initial bias. 
We hence propose a  model  on this line, wherein $w_{ig}^{(p)}$ is a monotone non-decreasing function of $w_{ii}^0 z_i^{(p-1)}$, and $w_{ib}^{(p)}$ is a monotone non-increasing function of $w_{ii}^0 z_i^{(p-1)}$.
Let node $i$ attribute a total of $\theta_i$ to the influence weights of the camps, that is,
$w_{ig}^{(p)}+w_{ib}^{(p)}=\theta_i$.
We propose the following natural model:

\begin{small}
\vspace{-3mm}
\begin{align}
\hspace{-3mm}
w_{ig}^{(p)} = \theta_i \left( \frac{1+w_{ii}^0 z_i^{(p-1)}}{2} \right)
\text{ , }
w_{ib}^{(p)} = \theta_i \left( \frac{1-w_{ii}^0 z_i^{(p-1)}}{2} \right)
\label{eqn:WonV}
\end{align}
\end{small}
\vspace{-4mm}

\noindent
So in the fund collection example, if in a phase, node $i$'s bias is very positive (which means that it is already aligned towards contributing significantly so as to help the good camp's cause), it would be easier for the good camp to influence node $i$ than it would be for the competing camp.


\paragraph{Camp Investments}
The good and bad camps attempt to directly influence the nodes so that their opinions are driven towards being positive and negative, respectively. 
We denote the investments made by the good and bad camps on node $i$ in phase $p$ by $x_i^{(p)}$ and $y_i^{(p)}$ respectively ($x_i^{(p)},y_i^{(p)} \geq 0,  \forall {i \in V}$ for $p =1,2$). 
In the fund collection example, the investments by a camp on a node could be in the form of effort and time spent in presenting convincing arguments in favor of the cause supported by the camp.
Since the influence of good camp (positive opinion) on node $i$ in phase $p$ would be an increasing function of both $x_i^{(p)}$ and $w_{ig}^{(p)}$, we assume the influence to be $+w_{ig}^{(p)} x_i^{(p)}$ (maintaining the linearity of Friedkin-Johnsen model). Similarly, $-w_{ib}^{(p)} y_i^{(p)}$ is the influence of bad camp (negative opinion) on node $i$ in phase $p$. 
Let $k_g$ and $k_b$ be the respective budgets of the good and bad camps. Hence the  camps should invest in the two phases such that $\sum_{i \in V} ( x_i^{(1)} + x_i^{(2)} ) \leq k_g$ and $\sum_{i \in V} ( y_i^{(1)} + y_i^{(2)} ) \leq k_b$.

\paragraph{Matrix Forms}
Let $\mathbf{W}$ be the matrix consisting of  weights $w_{ij}$ for each pair of nodes $(i,j)$.
Let $\mathbf{z^0}$, $\mathbf{w^0}$, 
${\Theta}$,
$\mathbf{w_g}^{\mathbf{(p)}}$, $\mathbf{w_b}^{\mathbf{(p)}}$, $\mathbf{x^{(p)}}$, $\mathbf{y^{(p)}}$, $\mathbf{z^{(p)}}$
be the vectors consisting of  elements
$z_i^0$, $w_{ii}^0$, $\theta_i$, $w_{ig}^{(p)}$, $w_{ib}^{(p)}$, $x_i^{(p)}$, $y_i^{(p)}$, $z_i^{(p)}$, respectively.
Let  operation $\circ$ denote Hadamard vector product,
that is, $(\mathbf{e}\circ\mathbf{f})_{i} = {e}_i {f}_i$.

\paragraph{Opinion Update Rule} 

Recall that the  update rule in  Friedkin-Johnsen model is given by

\vspace{-8mm}
\begin{small}
\begin{align*}
\forall {i \in V}:
z_i
\leftarrow w_{ii}^{0} z_i^0
+ \sum_{j \in V} w_{ij}z_j
\end{align*}
\end{small}
\vspace{-3mm}

Extending to multiple phases, the update rule in the $p^{\text{th}}$ phase is

\vspace{-4mm}
\begin{small}
\begin{align*}
\forall {i \in V}:
z_i^{(p)}
\leftarrow w_{ii}^{0} z_i^{(p-1)}
+ \sum_{j \in V} w_{ij}z_j^{(p)}
\end{align*}
\end{small}
\vspace{-3mm}

Accounting for camps' investments,
we get
 
\begin{small}
\vspace{-3mm}
\begin{align}
\forall {i \in V}:\;
z_i^{(p)}
&\leftarrow w_{ii}^{0} z_i^{(p-1)}
+ \sum_{j \in V} w_{ij}z_j^{(p)}
+ w_{ig}^{(p)} x_i^{(p)}
- w_{ib}^{(p)} y_i^{(p)}
\nonumber
\end{align}
\vspace{-3mm}
\end{small}

\begin{small}
\vspace{-4mm}
\begin{align}
\hspace{-5mm}
\Longleftrightarrow 
\mathbf{z}^{\mathbf{(p)}} & \leftarrow   \mathbf{w^0}   \circ   \mathbf{z^{\mathbf{(p-1)}}}  +  \mathbf{W}\mathbf{z}^{\mathbf{(p)}}   +  \mathbf{w_g}^{\mathbf{(p)}}   \circ   \mathbf{x}^{\mathbf{(p)}}  -  \mathbf{w_b}^{\mathbf{(p)}}   \circ   \mathbf{y}^{\mathbf{(p)}}   
\label{eqn:update_rule_matrix}
\end{align}
\vspace{-4mm}
\end{small}

\noindent
In phase $p$,
{the vectors $\mathbf{x}^{\mathbf{(p)}}, \mathbf{y}^{\mathbf{(p)}}, \mathbf{z^{\mathbf{(p-1)}}}$ 
stay unchanged};
the weights $\mathbf{w_g}^{\mathbf{(p)}}, \mathbf{w_b}^{\mathbf{(p)}}$ (which depend on $\mathbf{z^{\mathbf{(p-1)}}}$) and $\mathbf{w^0}$
also stay unchanged; while
$\mathbf{z}^{\mathbf{(p)}}$ gets updated.
{Hence, writing the update rule as recursion (with iterating integer $\tau \geq  0$), we get}

\vspace{-8mm}
\begin{small}
\begin{align*}
\mathbf{z}_{\langle\tau\rangle}^{\mathbf{(p)}} &= \mathbf{W}\mathbf{z}_{\langle\tau-1\rangle}^{\mathbf{(p)}} + \mathbf{w^0}\circ\mathbf{z^{\mathbf{(p-1)}}} + \mathbf{w_g}^{\mathbf{(p)}}\circ\mathbf{x}^{\mathbf{(p)}} - \mathbf{w_b}^{\mathbf{(p)}}\circ\mathbf{y}^{\mathbf{(p)}}
\end{align*}
\vspace{-4mm}
\end{small}

\noindent
Solving the recursion simplifies it to

\vspace{-3mm}
\begin{small}
\begin{align*}
\hspace{-2mm}
\mathbf{z}_{\langle\tau\rangle}^{\mathbf{(p)}}   =  
\mathbf{W}^\tau \mathbf{z}_{\langle 0 \rangle}^{\mathbf{(p)}}   +   \big( \sum_{\eta=0}^{\tau-1}  \mathbf{W}^\eta \big)(\mathbf{w^0}   \circ   \mathbf{z^{\mathbf{(p-1)}}}   & +   \mathbf{w_g}^{\mathbf{(p)}}   \circ   \mathbf{x}^{\mathbf{(p)}}    -   \mathbf{w_b}^{\mathbf{(p)}}   \circ   \mathbf{y}^{\mathbf{(p)}}) 
\end{align*}
\vspace{-2mm}
\end{small}

\noindent
Now, the initial bias for phase $p$ : $\mathbf{z}_{\langle 0 \rangle}^{\mathbf{(p)}} = \mathbf{z}^{\mathbf{(p-1)}}$.
Also, since $\sum_{j \in V} |w_{ij}| < 1$, we have that $\mathbf{W}$ is a strictly substochastic matrix (sum of each row strictly less than 1). Hence
its spectral radius is less than 1.
So when $\tau \rightarrow \infty$, we have $\lim_{\tau \rightarrow \infty} \mathbf{W}^\tau  = \mathbf{0}$ and
$\lim_{\tau \rightarrow \infty} \sum_{\eta=0}^{\tau-1} \mathbf{W}^\eta = (\mathbf{I}-\mathbf{W})^{-1}$
\cite{grabisch2017strategic}.
Hence,

\begin{small}
\vspace{-1mm}
\begin{align*}
\lim_{\tau \rightarrow \infty} \mathbf{z}_{\langle\tau\rangle}^{\mathbf{(p)}} = (\mathbf{I} - \mathbf{W})^{-1} (\mathbf{w^0} \circ \mathbf{z^{\mathbf{(p-1)}}}  +  \mathbf{w_g}^{\mathbf{(p)}} \circ \mathbf{x}^{\mathbf{(p)}}  -  \mathbf{w_b}^{\mathbf{(p)}} \circ \mathbf{y}^{\mathbf{(p)}})
\end{align*}
\vspace{-1mm}
\end{small}

\noindent
which is a constant vector.
So the  dynamics in phase $p$ converges to 
the steady state

\begin{small}
\vspace{-6mm}
\begin{align}
\hspace{-2mm}
\mathbf{z}^{\mathbf{(p)}}  =  (\mathbf{I} - \mathbf{W})^{-1} (\mathbf{w^0} \circ \mathbf{z^{(p-1)}}   +  \mathbf{w_g}^{\mathbf{(p)}} \circ \mathbf{x^{(p)}}  -  \mathbf{w_b}^{\mathbf{(p)}} \circ \mathbf{y^{(p)}})
\label{eqn:basic}
\end{align}
\vspace{-3mm}
\end{small}

\subsection{Formulation of Two-Phase Objective Function}
\label{sec:multiphase}

We now derive the objective function $\sum_{i \in V} z_i^{(2)}$, the sum of opinion values of all the nodes at the end of the second phase.
Premultiplying Equation~(\ref{eqn:basic}) by $\mathbf{1}^T$ gives

\begin{small}
\vspace{-6mm}
\begin{align*}
\hspace{-1mm}
\mathbf{1}^T \mathbf{z}^{\mathbf{(p)}} = \mathbf{1}^T (\mathbf{I}-\mathbf{W})^{-1} (\mathbf{w^0}\circ\mathbf{z^{(p-1)}} + \mathbf{w_g}\circ\mathbf{x^{(p)}} - \mathbf{w_b}\circ\mathbf{y^{(p)}})
\end{align*}
\vspace{-3mm}
\end{small}

\noindent
Let 
 $\Delta = (\mathbf{I}-\mathbf{W})^{-1}$ and
 $\mathbf{r}^T = \mathbf{1}^T (\mathbf{I}-\mathbf{W})^{-1}$,
 that is, {$r_i=\sum_{j \in V} \Delta_{ji}$}.
Since $\Delta = \sum_{\eta=0}^\infty \mathbf{W}^\eta$, we have that $\Delta_{ji}$ is the influence that  $j$ receives from  $i$ through walks of all possible lengths.
So $r_i=\sum_{j \in V} \Delta_{ji}$ can be viewed as  overall influencing power of  $i$.
Substituting these in the above equation, we get

\begin{small}
\vspace{-2mm}
\begin{align*}
\sum_{i \in V} z_i^{(p)} = \sum_{i \in V} r_i ( w_{ii}^0 z_i^{(p-1)} +  w_{ig}x_i^{(p)} - w_{ib}y_i^{(p)} ) 
\end{align*}
\vspace{-2mm}
\end{small}

\noindent
{When $p=1$, this is the sum of opinions at the end of the first phase:}

\begin{small}
\vspace{-3mm}
\begin{align}
\sum_{i \in V} z_i^{(1)} = \sum_{i \in V} r_i ( w_{ii}^0 z_i^{0} +  w_{ig}^{(1)} x_i^{(1)} - w_{ib}^{(1)} y_i^{(1)} ) 
\label{eqn:sumatP_1}
\end{align}
\vspace{-2mm}
\end{small}

\noindent
Similarly, when $p=2$, the sum of opinion values at the end of the second phase is (by also using Equations~(\ref{eqn:WonV})):

\begin{small}
\vspace{-3mm}
\begin{align}
&\sum_{j \in V} z_j^{(2)} 
= \sum_{j \in V} r_j ( w_{jj}^0 z_j^{(1)} + w_{jg}^{(2)} x_j^{(2)} - w_{jb}^{(2)} y_j^{(2)} )
\nonumber
\\
&= \sum_{j \in V} r_j \left( w_{jj}^0 z_j^{(1)} + \frac{\theta_j}{2} ( {1 + w_{jj}^0 z_j^{(1)}} ) x_j^{(2)} -  \frac{\theta_j}{2} ( {1 - w_{jj}^0 z_j^{(1)}} ) y_j^{(2)} \right)
\nonumber
\\
&= \sum_{j \in V} r_j w_{jj}^0 z_j^{(1)} \left( 1  +  \frac{\theta_j}{2} x_j^{(2)}  +  \frac{\theta_j}{2} y_j^{(2)} \right) + \sum_{j \in V} r_j \frac{\theta_j}{2} (x_j^{(2)}  -  y_j^{(2)})
\label{eqn:pausehere}
\end{align}
\vspace{-1mm}
\end{small}

\noindent
{The first term $\sum_{j \in V} r_j w_{jj}^0 z_j^{(1)} \left( 1 + \frac{\theta_j}{2} x_j^{(2)} + \frac{\theta_j}{2} y_j^{(2)} \right)$
 can be} obtained 
by 
premultiplying (\ref{eqn:sumatP_1}) by $(\mathbf{r} \circ \mathbf{w^0} \circ (\mathbf{1} + \frac{\Theta}{2}\circ \mathbf{x^{(2)}} + \frac{\Theta}{2}\circ \mathbf{y^{(2)}}))^T$ and using (\ref{eqn:WonV}).
Hence,

\begin{small}
\begin{align*}
&\sum_{j\in V} r_j w_{jj}^0 z_j^{(1)} \left( 1 + \frac{\theta_j}{2} x_j^{(2)} + \frac{\theta_j}{2} y_j^{(2)} \right)
\\
&= \left( \mathbf{r} \circ \mathbf{w^0} \circ (\mathbf{1} + \frac{\Theta}{2}\circ \mathbf{x^{(2)}} + \frac{\Theta}{2}\circ \mathbf{y^{(2)}}) \right)^T
(\mathbf{I}-\mathbf{W})^{-1}
(\mathbf{w^0}\circ\mathbf{z^{(0)}} 
+\mathbf{w_g^{(1)}}\circ\mathbf{x^{(1)}} - \mathbf{w_b^{(1)}}\circ\mathbf{y^{(1)}}) 
\\
&= \sum_{i\in V} \left( \sum_{j\in V} r_j w_{jj}^0 \left(1+\frac{\theta_j}{2} x_j^{(2)} + \frac{\theta_j}{2} y_j^{(2)}\right) \Delta_{ji} \right)
( w_{ii}^0 z_i^{0} 
 + w_{ig}^{(1)} x_i^{(1)} - w_{ib}^{(1)} y_i^{(1)} )
\\
&= \sum_{i\in V} \left( \sum_{j\in V} r_j w_{jj}^0 \left(1+\frac{\theta_j}{2} x_j^{(2)} + \frac{\theta_j}{2} y_j^{(2)}\right) \Delta_{ji} \right)
\Bigg( w_{ii}^0 z_i^{0}  
+ \frac{\theta_i}{2} ( {1+w_{ii}^0 z_i^{0}} ) x_i^{(1)} - \frac{\theta_i}{2} ( {1-w_{ii}^0 z_i^{0}} ) y_i^{(1)} \Bigg)
\\
&= \sum_{i\in V} \sum_{j\in V} \left( w_{ii}^0 z_i^{0}
\left(1+\frac{\theta_i}{2} x_i^{(1)} + \frac{\theta_i}{2} y_i^{(1)} \right) + \frac{\theta_i}{2}(x_i^{(1)}-y_i^{(1)}) \right)
\left( r_j w_{jj}^0 \Delta_{ji}  \left(1+\frac{\theta_j}{2} x_j^{(2)} + \frac{\theta_j}{2} y_j^{(2)}\right) \right)
\end{align*}
\end{small}
Substituting this in (\ref{eqn:pausehere}), we get we get that 
$\sum_{i \in V} z_i^{(2)}$ equals
\begin{small}
\begin{align}
\nonumber
\sum_{i\in V} \sum_{j\in V} \left( w_{ii}^0 z_i^{0}
\left(1+\frac{\theta_i}{2} x_i^{(1)} + \frac{\theta_i}{2} y_i^{(1)} \right) + \frac{\theta_i}{2}(x_i^{(1)}-y_i^{(1)}) \right)
\left( r_j w_{jj}^0 \Delta_{ji}  \left(1+\frac{\theta_j}{2} x_j^{(2)} + \frac{\theta_j}{2} y_j^{(2)}\right) \right)
\\
+ \sum_{j \in V} r_j \frac{\theta_j}{2} (x_j^{(2)}  -  y_j^{(2)})
\label{eqn:orig_dep_2camps}
\end{align}
\end{small}
\vspace{-3mm}

For notational simplicity, let $b_{ji} = r_j w_{jj}^0 \Delta_{ji}$, $c_i = w_{ii}^0 z_i^{0}$.
We  saw that $r_i=\sum_{j \in V} \Delta_{ji}$ indicates the influencing power of node $i$.
Now, $b_{ji} = \Delta_{ji} w_{jj}^0 r_j$ quantifies the overall influence $\Delta_{ji}$ of  $i$ on  $j$, which would give weightage $w_{jj}^0$ to its initial bias in the next phase, and have an influencing power of $r_j$ in the next phase.
Hence $b_{ji}$ can be interpreted as the influence of node $i$ on the network through node $j$, looking one phase ahead.
So, $\sum_{j \in V} b_{ji}$ can be viewed as the overall influencing power of node $i$, looking one phase ahead.
We denote $s_i=\sum_{j \in V} b_{ji} = \sum_{j \in V} r_j w_{jj}^0 \Delta_{ji}$.
In matrix notation, we have $\mathbf{s} = \Delta^T ( \mathbf{r} \circ \mathbf{w^0} )$.
We identify the role of parameter $s_i$ in our analysis later.

\paragraph{Two-phase Katz centrality}
Note that
$r_i = \left( (\mathbf{I}-\mathbf{W}^T)^{-1} \mathbf{1} \right)_i$ 
can be interpreted as a form of {\em Katz centrality} \cite{katz1953new} of node $i$.
The original Katz centrality is defined for an adjacency matrix $\mathcal{A}$ with all edges having the same weight $\alpha$.
Katz centrality of node $i$ is defined as the $i^{\text{th}}$ element of vector
$
\left( \left( \mathbf{I}-\alpha \mathcal{A}^T\right)^{-1} \! - \mathbf{I} \right) \mathbf{1}
=
 \left( \mathbf{I}-\alpha \mathcal{A}^T\right)^{-1} \mathbf{1}  - \mathbf{1}
$,
for $0\!<\!\alpha \!<\! \frac{1}{|\rho|}$ where $\rho$ is the largest eigenvalue of $\mathcal{A}$.
In our case where 
$
\mathbf{r}= \left( \mathbf{I}- \mathbf{W}^T\right)^{-1} \mathbf{1}
$,
 $\mathcal{A}$ is replaced by the weighted adjacency matrix $\mathbf{W}$, for which $|\rho|<1$ (since $\mathbf{W}$ is strictly substochastic), and we have $\alpha=1$. The subtraction of  vector $\mathbf{1}$ is common for all nodes, so its relative effect can be ignored.
Hence, $r_i$ measures node $i$'s relative influence in a social network when there are is no subsequent phase to follow. In the two-phase setting, this applies to the second phase since it is the terminal phase.
However,  while selecting optimal nodes in the first phase, when there is a subsequent phase to follow, the effectiveness of  node $i$ depends on its influencing power over those nodes ($j$), which would give good weightage ($w_{jj}^0$) to their initial biases in the second phase, as well as have a good influencing power over other nodes ($r_j$) in the second phase. This is precisely captured by $s_i$, and so it can be viewed as  the two-phase  Katz centrality.

Now, using this simplified notation, Equation (\ref{eqn:orig_dep_2camps}) can be written as

\begin{small}
\vspace{-3mm}
\begin{align}
\nonumber 
\sum_{i \in V} z_i^{(2)} 
= 
& \sum_{i \in V} \sum_{j \in V} c_i b_{ji} 
  +   \sum_{j \in V} x_j^{(2)} \frac{\theta_j}{2} \Big( \sum_{i \in V}  c_i b_{j i}   +   r_j \Big)
  +   \sum_{j \in V} y_j^{(2)} \frac{\theta_j}{2} \Big( \sum_{i \in V} c_i b_{j i}   -   r_j \Big)
\\ \hspace{-3mm}&\;\;\;\;
\nonumber
+ \sum_{i \in V} x_i^{(1)}  \frac{\theta_i}{2} (1   +   c_i) \Big( s_{i}
  +   \sum_{j \in V} x_j^{(2)} \frac{\theta_j}{2} b_{j i}
  +   \sum_{j \in V} y_j^{(2)} \frac{\theta_j}{2} b_{j i} \Big)
\\ \hspace{-3mm}&\;\;\;\;
- \sum_{i \in V} y_i^{(1)} \frac{\theta_i}{2} (1   -   c_i) \Big( s_i 
  +   \sum_{j \in V} x_j^{(2)} \frac{\theta_j}{2} b_{j i}
  +   \sum_{j \in V} y_j^{(2)} \frac{\theta_j}{2} b_{j i} \Big)
\label{eqn:raw2phase2camp}
\end{align}
\vspace{-4mm}
\end{small}

\subsection{The  Two-phase Investment Game}

Following our model, we have that $(\mathbf{x^{(1)}},\mathbf{x^{(2)}})$ is the strategy of the good camp for the two phases, and $(\mathbf{y^{(1)}},\mathbf{y^{(2)}})$ is the strategy of the bad camp.
Given an investment strategy profile $\big((\mathbf{x^{(1)}},\mathbf{x^{(2)}}),(\mathbf{y^{(1)}},\mathbf{y^{(2)}})\big)$,
let $u_g\big((\mathbf{x^{(1)}},\mathbf{x^{(2)}}),(\mathbf{y^{(1)}},\mathbf{y^{(2)}})\big)$ 
be  the good camp's utility
and $u_b\big((\mathbf{x^{(1)}},\mathbf{x^{(2)}}),(\mathbf{y^{(1)}},\mathbf{y^{(2)}})\big)$ be the bad camp's utility.
The good camp aims to maximize the value of  (\ref{eqn:raw2phase2camp}), while the bad camp aims to minimize it. So,

\begin{small}
\vspace{-4mm}
\begin{align}
\nonumber
u_g\big((\mathbf{x^{(1)}},\mathbf{x^{(2)}}),(\mathbf{y^{(1)}},\mathbf{y^{(2)}})\big)&= \sum_{i \in V} z_i^{(2)}
\\
\;\;\; \text{and} \;\;\;
u_b\big((\mathbf{x^{(1)}},\mathbf{x^{(2)}}),(\mathbf{y^{(1)}},\mathbf{y^{(2)}})\big)&= -\sum_{i \in V} z_i^{(2)}
\label{eqn:2phaseutilities}
\end{align}
\vspace{-4mm}
\end{small}

\noindent
with the following constraints on the investment strategies:

\begin{small}
\vspace{-4mm}
\begin{gather}
\nonumber
\sum_{i \in V} \big( x_i^{(1)} + x_i^{(2)} \big) \leq k_g \;\;,\;\; \sum_{i \in V} \big( y_i^{(1)} + y_i^{(2)} \big) \leq k_b
\\
\nonumber
\forall {i \in V}: x_i^{(1)} , x_i^{(2)} , y_i^{(1)} , y_i^{(2)} \geq 0
\end{gather}
\vspace{-4mm}
\end{small}

\noindent
The game can thus be viewed as a two-player zero-sum game,
where the players determine their investment strategies $(\mathbf{x^{(1)}},\mathbf{x^{(2)}})$ and $(\mathbf{y^{(1)}},\mathbf{y^{(2)}})$;
the good camp invests as per $\mathbf{x^{(1)}}$ in the first phase and as per $\mathbf{x^{(2)}}$ in the second phase, and 
the bad camp invests as per $\mathbf{y^{(1)}}$ in the first phase and as per $\mathbf{y^{(2)}}$ in the second phase.
Our objective essentially is to find the Nash equilibrium strategies of the two camps.

First, we consider a simplified yet interesting (and not-yet-studied-in-literature) case where  budget of one of the camps is $0$ (say $k_b=0$); so effectively
we have only the good camp. 

\vspace{-2mm}
\section{The Non-Competitive Case}
\label{sec:dep_onecamp}

\vspace{-1mm}
For the non-competitive case, when there is only one camp (say the good camp, without loss of analytical generality), we  have $y_i^{(1)} = y_i^{(2)} = 0, \forall {i \in V}$ in
Equation~(\ref{eqn:raw2phase2camp}).
So we get

\begin{small}
\vspace{-2mm}
\begin{align}
\hspace{-4mm}
\sum_{i \in V} z_i^{(2)} 
&= 
 \sum_{i \in V} \sum_{j \in V} c_i b_{ji} 
  +   \sum_{j \in V} x_j^{(2)} \frac{\theta_j}{2} \Big( \sum_{i \in V}  c_i b_{j i}   +   r_j \Big)
+ \sum_{i \in V} x_i^{(1)}  \frac{\theta_i}{2} (1   +   c_i) \Big( s_{i}
  +   \sum_{j \in V} x_j^{(2)} \frac{\theta_j}{2} b_{j i}
  \Big)
  \nonumber
  \\
  \nonumber
&=  \sum_{i \in V} \sum_{j \in V} c_i b_{ji} 
   + \sum_{i \in V} x_i^{(1)}  \frac{\theta_i}{2} (1   +   c_i)  s_{i}
    +   \sum_{j \in V} x_j^{(2)} \frac{\theta_j}{2} \Big( \sum_{i \in V}  c_i b_{j i}   +   r_j \Big)
  \\ &\;\;\;\;\;\;\;\;\;\;\;\;\;\;\;\;\;\;\;\;\;\;\;\;\;\;\;\;\;\;\;\;\;\;\;\;\;\;\;\;\;\;\;\;\;\;\;\;\;\;\;\;
    + \sum_{i \in V} x_i^{(1)}  \frac{\theta_i}{2} (1   +   c_i)  \sum_{j \in V} x_j^{(2)} \frac{\theta_j}{2} b_{j i}
      \label{eqn:raw2phase}
    \\
&= \sum_{i \in V} x_i^{(1)}  \frac{\theta_i}{2} (1   +   c_i) \Big( s_{i}
  +   \sum_{j \in V} x_j^{(2)} \frac{\theta_j}{2} b_{j i}
  \Big)  + \text{(terms independent of $x_i^{(1)}$)}
   \label{eqn:raw2phasefix2}
  \\
  &=  \sum_{j \in V} x_j^{(2)} \frac{\theta_j}{2} \Big(  r_j +  \sum_{i \in V}  c_i b_{j i}  + \sum_{i \in V} x_i^{(1)}  \frac{\theta_i}{2} (1   +   c_i) b_{ji} \Big) + \text{(terms independent of $x_j^{(2)}$)}
   \label{eqn:raw2phasefix1}
\end{align}
\vspace{-3mm}
\end{small}

\noindent
Deducing from Equations (\ref{eqn:raw2phasefix2}) and (\ref{eqn:raw2phasefix1}) that $\sum_{i \in V} z_i^{(2)}$
 is a bilinear function in $\mathbf{x^{(1)}}$ and $\mathbf{x^{(2)}}$, we prove our next result.

Let the budget $k_g$ be split such that $k_g^{(1)}$ and $k_g^{(2)}$ are the first and second phase investments, respectively.
Our method of finding optimal $k_g^{(1)}$ and $k_g^{(2)}$ is to first search for them in the search space $k_g^{(1)}+k_g^{(2)} \in (0, k_g]$ and then compare the thus obtained value of $\sum_{i \in V} z_i^{(2)}$ with that corresponding to $k_g^{(1)}=k_g^{(2)}=0$, so as to get the optimal value of $\sum_{i \in V} z_i^{(2)}$.

\begin{proposition}
In the search space $k_g^{(1)}+k_g^{(2)} \in (0,k_g]$,
it is optimal  for the good camp  to exhaust its entire budget $(k_g^{(1)}+k_g^{(2)}=k_g)$, 
and to invest on at most one node in each phase.
\label{prop:allornothing}
\end{proposition}
\begin{proof}
Given any $\mathbf{x^{(2)}}$, Expression (\ref{eqn:raw2phasefix2}) can be maximized w.r.t. $\mathbf{x^{(1)}}$ by allocating $k_g - \sum_{j \in V} x_j^{(2)}$ to a single node $i$ that maximizes
$
\frac{\theta_i}{2} (1   +   c_i) \left( s_{i}
  +   \sum_{j \in V} x_j^{(2)} \frac{\theta_j}{2} b_{j i}
  \right) 
  $,
  if this value is positive.
  In case of multiple such nodes, one node can be chosen at random.
  If this value is non-positive for all  nodes, it is optimal to have $\mathbf{x^{(1)}} = \mathbf{0}$.
  When $\mathbf{x^{(1)}} = \mathbf{0}$, 
  Expression (\ref{eqn:raw2phasefix1}) now implies that it is optimal to allocate the entire budget $k_g$ in  second phase to a single node $j$ that maximizes
$
    \frac{\theta_j}{2} \left( r_j + \sum_{i \in V} c_i b_{ji}  \right)
$,
      if this value is positive. 
      If this value is non-positive for all  nodes, it is optimal to have $\mathbf{x^{(2)}} = \mathbf{0}$.
      This is the case where starting with an $\mathbf{x^{(2)}}$, we conclude that it is optimal to either invest $k_g - \sum_{j \in V} x_j^{(2)}$ on a single node in  first phase, or invest the entire budget $k_g$ on a single node in  second phase, or invest in neither phase.
      
      Similarly using (\ref{eqn:raw2phasefix1}),
starting with an $\mathbf{x^{(1)}}$, we can conclude that it is optimal to either invest $k_g - \sum_{j \in V} x_j^{(1)}$ on a single node in the second phase, or invest the entire $k_g$ on a single node in the first phase, or invest in neither phase.

So starting from any $\mathbf{x^{(1)}}$ or $\mathbf{x^{(2)}}$, we can iteratively improve (need not be strictly) on the value of (\ref{eqn:raw2phase}) by investing on at most one node in a given phase.
Furthermore, it is suboptimal to have $k_g^{(1)}+k_g^{(2)}<k_g$ unless $k_g^{(1)}=k_g^{(2)}=0$. 

Note that, in case of multiple  nodes maximizing the respective coefficients
(for instance, 
$
\frac{\theta_i}{2} (1   +   c_i) \left( s_{i}
  +   \sum_{j \in V} x_j^{(2)} \frac{\theta_j}{2} b_{j i}
  \right) 
  $ while determining optimal
  $\mathbf{x^{(1)}}$)
  we mentioned that 
   one such node can be chosen at random.
   So, there may be  optimal strategies in which the camp could invest on multiple nodes in a phase.
   However, since investing on one node per phase suffices to achieve the optimum, it is an optimal strategy (not the only optimal strategy) to invest on at most one node in each phase.
\end{proof}

Following Proposition \ref{prop:allornothing}, there exist optimal vectors $\mathbf{x^{(1)}}$ and $\mathbf{x^{(2)}}$ that maximize (\ref{eqn:raw2phase}), such that $x_{\alpha}^{(1)} = k_g^{(1)}, x_{\beta}^{(2)} = k_g^{(2)} , x_{i \neq \alpha}^{(1)} = x_{j \neq \beta}^{(2)} = 0$.
Now the next step is to find  nodes $\alpha$ and $\beta$ that maximize (\ref{eqn:raw2phase}).
By incorporating $\alpha$ and $\beta$ in (\ref{eqn:raw2phase}), 
we have  
$\sum_{i \in V} z_i^{(2)}$ equals

\begin{small}
\vspace{-1mm}
\begin{align}
\nonumber
& \sum_{i \in V} \sum_{j \in V} c_i b_{ji} +  k_g^{(1)}  \frac{\theta_{\alpha}}{2} (1+c_{\alpha}) s_{\alpha}
  + k_g^{(2)} \frac{\theta_{\beta}}{2}  \Big(\sum_{i \in V} c_i b_{\beta i} + r_{\beta}\Big)
  \\ &\;\;\;\;\;\;\;\;\;\;\;\;\;\;\;\;\;\;\;\;\;\;\;\;\;\;\;\;\;\;\;\;\;\;\;\;\;\;\;\;\;\;\;\;\;\;\;\;\;\;
  + k_g^{(1)}  k_g^{(2)} \frac{\theta_{\alpha}\theta_{\beta}}{4}  (1+c_{\alpha}) b_{\beta \alpha}
  \label{eqn:twophase_onecamp_1}
\end{align}
\end{small}
\vspace{-2mm}

Now, for a given pair $(\alpha,\beta)$, we will find the optimal values of $k_g^{(1)}$ and $k_g^{(2)}$ from (\ref{eqn:twophase_onecamp_1}). From Proposition \ref{prop:allornothing}, we first consider $k_g^{(2)} = k_g - k_g^{(1)}$. So the expression to be maximized is

\begin{small}
\vspace{-1mm}
\begin{align*}
& \sum_{i \in V} \sum_{j \in V} c_i b_{ji} +  k_g^{(1)}  \frac{\theta_{\alpha}}{2} (1+c_{\alpha}) s_{\alpha}
  + (k_g   -   k_g^{(1)}) \frac{\theta_{\beta}}{2}  \Big(\sum_{i \in V} c_i b_{\beta i} + r_{\beta}\Big)
  \\ &\;\;\;\;\;\;\;\;\;\;\;\;\;\;\;\;\;\;\;\;\;\;\;\;\;\;\;\;\;\;\;\;\;\;\;\;\;\;\;\;\;\;\;\;\;\;\;\;\;\;
  + k_g^{(1)}  (k_g   -   k_g^{(1)}) \frac{\theta_{\alpha}\theta_{\beta}}{4}  (1+c_{\alpha}) b_{\beta \alpha}
\end{align*}
\end{small}
\vspace{-2mm}

\noindent
Equating its first derivative w.r.t. $k_g^{(1)}$ to zero, we get the candidate $k_g^{(1)}$ for pair $(\alpha,\beta)$ to be

\begin{small}
\vspace{-2.5mm}
\begin{align*}
\frac{k_g}{2} + \frac{ s_{\alpha}}{\theta_{\beta}b_{\beta \alpha}} - \frac{\sum_{i \in V} b_{\beta i}c_i + r_{\beta}}{\theta_{\alpha}b_{\beta \alpha}(1+c_{\alpha})}
\end{align*}
\vspace{-1.5mm}
\end{small}

\noindent
A valid value of $k_g^{(1)}$ can be obtained only if 
the denominators in  above expression are non-zero. 
However, a zero denominator would mean that Expression (\ref{eqn:twophase_onecamp_1}) is linear, resulting in the optimal $k_g^{(1)}$ being either $0$ or $k_g$.
Also, if the second derivative with respect to $k_g^{(1)}$ is positive, that is, 
$- \theta_{\alpha}\theta_{\beta} (1+c_{\alpha}) b_{\beta \alpha} > 0$,
optimal $k_g^{(1)} $ is either $0$ or $k_g$.
If the second derivative with respect to $k_g^{(1)}$ is negative: $- \theta_{\alpha}\theta_{\beta} (1+c_{\alpha}) b_{\beta \alpha}   =  - \theta_{\alpha}\theta_{\beta} r_\beta w_{\beta\beta}^0 \Delta_{\beta\alpha} (w_{\alpha\alpha}^0 z_\alpha^{0} + 1)  <  0$,
and since $k_g^{(1)}$ is bounded in $[0,k_g]$,  optimal $k_g^{(1)}$ for  pair $(\alpha,\beta)$ is
(since $c_i=w_{ii}^0 z_i^{0}$, $b_{ji}=r_j w_{jj}^0 \Delta_{ji}$, $s_i=\sum_{j \in V} b_{ji}$):

\begin{small}
\vspace{-2mm}
\begin{align}
\hspace{-3mm}
\min   \left\{   \max   \bigg\{  \frac{k_g}{2}   +   \frac{s_{\alpha}}{\theta_{\beta} r_{\beta} w_{\beta \beta}^0 \Delta_{\beta \alpha}}   -   \frac{1+w_{\beta \beta}^0 \sum_{i \in V}  \Delta_{\beta i} w_{ii}^0 z_i^{0}}{\theta_{\alpha}  w_{\beta \beta}^0 \Delta_{\beta \alpha} (1   +   w_{\alpha \alpha}^0 z_{\alpha}^{0})} , 0  \bigg\}  ,   k_g   \right\}
\label{eqn:optk1_new}
\end{align}
\vspace{-4mm}
\end{small}

\noindent
and the corresponding optimal value of $k_g^{(2)}$ for  pair $(\alpha,\beta)$  is

\begin{small}
\vspace{-2mm}
\begin{align}
\hspace{-3mm}
\min   \left\{   \max   \bigg\{ \frac{k_g}{2}   -   \frac{s_{\alpha}}{\theta_{\beta} r_{\beta} w_{\beta \beta}^0 \Delta_{\beta \alpha}}   +   \frac{1+w_{\beta \beta}^0 \sum_{i \in V}  \Delta_{\beta i} w_{ii}^0 z_i^{0}}{\theta_{\alpha}  w_{\beta \beta}^0 \Delta_{\beta \alpha} (1   +   w_{\alpha \alpha}^0 z_{\alpha}^{0})} , 0  \bigg\} , k_g   \right\}
\label{eqn:optk2_new}
\end{align}
\vspace{-3mm}
\end{small}

When we assumed $k_g^{(1)}$ and $k_g^{(2)}$ to be fixed, we had to iterate through all $(\alpha,\beta)$ pairs to determine the one that gives the optimal value of Expression (\ref{eqn:twophase_onecamp_1}). Now, whenever we look at an $(\alpha,\beta)$ pair, we can determine the corresponding optimal values of $k_g^{(1)}$ and $k_g^{(2)}$ using (\ref{eqn:optk1_new}) and (\ref{eqn:optk2_new}), and hence determine the value of Expression (\ref{eqn:twophase_onecamp_1}) by plugging in the optimal $k_g^{(1)}$ and $k_g^{(2)}$ and that $(\alpha,\beta)$ pair. 
The optimal pair $(\alpha,\beta)$ can thus be obtained as the pair that maximizes (\ref{eqn:twophase_onecamp_1}).

Above analysis holds when $k_g^{(1)}+k_g^{(2)}=k_g$.
From Proposition~\ref{prop:allornothing}, we need to consider one more possibility that 
$k_g^{(1)}=k_g^{(2)}=0$, which gives a constant value $\sum_{i \in V} \sum_{j \in V} c_i b_{ji}$ for Expression (\ref{eqn:twophase_onecamp_1}).
Let $(0,0)$ correspond to this additional possibility.
It is hence optimal to invest $k_g^{(1)}$ (obtained using (\ref{eqn:optk1_new})) on node $\alpha$ in the first phase and $k_g^{(2)}$ (obtained using (\ref{eqn:optk2_new})) on node $\beta$ in the second phase, subject to it giving a  value greater than $\sum_{i \in V} \sum_{j \in V} c_i b_{ji}$ to Expression (\ref{eqn:twophase_onecamp_1}).

Since we iterate through $(n^2+1)$ possibilities (namely, $(\alpha,\beta) \in V \times V \cup \{(0,0)\}$), the above procedure gives a polynomial time algorithm for determining the optimal budget split and the optimal investments on nodes in two phases.
It can also be shown that the complexity of this algorithm is dominated by the matrix inversion operation (computation of $\Delta$).

\begin{remark} 
\label{rem:singlecamp}
For non-negative values of parameters,
 (\ref{eqn:optk1_new}) indicates that for a given $(\alpha,\beta)$ pair, the good camp would want to invest more in the first phase for a higher $s_\alpha$. This is intuitive from our understanding of $s_\alpha$ being viewed as the influencing power of node $\alpha$ looking one phase ahead. Similarly,   (\ref{eqn:optk2_new}) indicates that it would want to invest more in  second phase for a higher $r_\beta$, since $r_\beta$ can be viewed as the influencing power of node $\beta$ in the immediate phase. 
Also, (\ref{eqn:optk1_new}) and (\ref{eqn:optk2_new}) indicate that a higher $\theta_\alpha$ drives the camp to invest in  first phase and a higher $\theta_\beta$ drives it to invest in  second phase.
Since $w_{ig}$ is an increasing function of $\theta_i$, this implicitly means that a node with a higher $w_{ig}$ drives the good camp to invest in the phase in which that node is selected.
Further, we  illustrate the role of $w_{ii}^0$ using simulations in Section \ref{sec:ODSNmultiphase_sim}.
\end{remark}

\section{The Case of Competing Camps
}
\label{sec:dep_2camps}

We now analyze the  general scenario  involving the two competing camps.
We first prove the following result.

\begin{proposition}
In the search space $k_g^{(1)}+k_g^{(2)} \in (0,k_g]$,
it is optimal  for the good camp  to have $k_g^{(1)}+k_g^{(2)}=k_g$, 
and to invest on at most one node in each phase.
In the search space $k_b^{(1)}+k_b^{(2)} \in (0,k_b]$,
it is optimal  for the bad camp  to have $k_b^{(1)}+k_b^{(2)}=k_b$, 
and to invest on at most one node in each phase.
\label{prop:allornothing2camps}
\end{proposition}

\begin{proof}
We  show that $\sum_{i \in V} z_i^{(2)}$ is a multilinear function, since it can be written as a linear function  in $\mathbf{x^{(1)}},\mathbf{y^{(1)}},\mathbf{x^{(2)}},\mathbf{y^{(2)}}$ individually, as follows:

%

\begin{small}
\vspace{-2mm}
\begin{align*}
\nonumber 
\sum_{i \in V} z_i^{(2)} 
&= 
 \sum_{i \in V} \sum_{j \in V} c_i b_{ji} 
  +   \sum_{j \in V} x_j^{(2)} \frac{\theta_j}{2} \Big( \sum_{i \in V}  c_i b_{j i}   +   r_j \Big)
  +   \sum_{j \in V} y_j^{(2)} \frac{\theta_j}{2} \Big( \sum_{i \in V} c_i b_{j i}   -   r_j \Big)
\\ 
&\;\;\;\;\;\;\;\;\;\;\;\;\;\;\;\;\;\;\;\;\;\;\;\;\;\;\;\;
\nonumber
+ \sum_{i \in V} x_i^{(1)}  \frac{\theta_i}{2} (1   +   c_i) \Big( s_{i}
  +   \sum_{j \in V} x_j^{(2)} \frac{\theta_j}{2} b_{j i}
  +   \sum_{j \in V} y_j^{(2)} \frac{\theta_j}{2} b_{j i} \Big)
\\ 
&\;\;\;\;\;\;\;\;\;\;\;\;\;\;\;\;\;\;\;\;\;\;\;\;\;\;\;\;
- \sum_{i \in V} y_i^{(1)} \frac{\theta_i}{2} (1   -   c_i) \Big( s_i 
  +   \sum_{j \in V} x_j^{(2)} \frac{\theta_j}{2} b_{j i}
  +   \sum_{j \in V} y_j^{(2)} \frac{\theta_j}{2} b_{j i} \Big)
  \\
  &=
  \sum_{i \in V} x_i^{(1)}  \frac{\theta_i}{2} (1   +   c_i) \Big( s_{i}
    +   \sum_{j \in V} x_j^{(2)} \frac{\theta_j}{2} b_{j i}
    +   \sum_{j \in V} y_j^{(2)} \frac{\theta_j}{2} b_{j i} \Big)
    \\
    &\;\;\;\;\;\;\;\;\;\;\;\;\;\;\;\;\;\;\;\;\;\;\;\;\;\;\;\;\;\;\;\;\;\;\;\;\;\;\;\;\;\;\;\;\;\;\;\;\;\;\;\;\;\;\;\;\;\;\;\;\;\;\;\;\;\;\;\;\;\;\;\;\;\;
    + \text{(terms independent of $x_i^{(1)}$)}
    \\
    &=
    - \sum_{i \in V} y_i^{(1)} \frac{\theta_i}{2} (1   -   c_i) \Big( s_i 
      +   \sum_{j \in V} x_j^{(2)} \frac{\theta_j}{2} b_{j i}
      +   \sum_{j \in V} y_j^{(2)} \frac{\theta_j}{2} b_{j i} \Big)
      \\
      &\;\;\;\;\;\;\;\;\;\;\;\;\;\;\;\;\;\;\;\;\;\;\;\;\;\;\;\;\;\;\;\;\;\;\;\;\;\;\;\;\;\;\;\;\;\;\;\;\;\;\;\;\;\;\;\;\;\;\;\;\;\;\;\;\;\;\;\;\;\;\;\;\;\;
      + \text{(terms independent of $y_i^{(1)}$)}
      \\
      &=
      \sum_{j \in V} x_j^{(2)} \frac{\theta_j}{2} \Big( r_j +  \sum_{i \in V}  c_i b_{j i}   +   \sum_{i \in V} x_i^{(1)}  \frac{\theta_i}{2} (1   +   c_i) b_{ji} - \sum_{i \in V} y_i^{(1)} \frac{\theta_i}{2} (1   -   c_i) b_{ji} \Big)
      \\
&\;\;\;\;\;\;\;\;\;\;\;\;\;\;\;\;\;\;\;\;\;\;\;\;\;\;\;\;\;\;\;\;\;\;\;\;\;\;\;\;\;\;\;\;\;\;\;\;\;\;\;\;\;\;\;\;\;\;\;\;\;\;\;\;\;\;\;\;\;\;\;\;\;\;
      + \text{(terms independent of $x_j^{(2)}$)}
      \\
      &=
      - \sum_{j \in V} y_j^{(2)} \frac{\theta_j}{2} \Big( r_j - \sum_{i \in V} c_i b_{j i}  -   \sum_{i \in V} x_i^{(1)}  \frac{\theta_i}{2} (1   +   c_i) b_{ji} + \sum_{i \in V} y_i^{(1)} \frac{\theta_i}{2} (1   -   c_i) b_{ji} \Big)
      \\
      &\;\;\;\;\;\;\;\;\;\;\;\;\;\;\;\;\;\;\;\;\;\;\;\;\;\;\;\;\;\;\;\;\;\;\;\;\;\;\;\;\;\;\;\;\;\;\;\;\;\;\;\;\;\;\;\;\;\;\;\;\;\;\;\;\;\;\;\;\;\;\;\;\;\;
            + \text{(terms independent of $y_j^{(2)}$)}
\end{align*}
\vspace{-3mm}
\end{small}

The rest of the proof follows on similar lines as Proposition \ref{prop:allornothing}.
(Note the negative signs assigned to the coefficients $y_i^{(1)}$ and $y_j^{(2)}$ since the bad camp computes the values of these parameters so as to minimize $\sum_{i \in V} z_i^{(2)}$).
\end{proof}

From Proposition \ref{prop:allornothing2camps}, there exist optimal vectors $\mathbf{x^{(1)}},\mathbf{x^{(2)}}$ for  good camp and optimal vectors $\mathbf{y^{(1)}},\mathbf{y^{(2)}}$ for  bad camp, 
{such that $x_{\alpha}^{(1)}  =  k_g^{(1)}, x_{\beta}^{(2)}  =  k_g^{(2)} , y_{\gamma}^{(1)}  =  k_b^{(1)} , y_{\delta}^{(2)}  =  k_b^{(2)}$, and} {$x_{i \neq \alpha}^{(1)}  =  x_{j \neq \beta}^{(2)}  =  y_{i \neq \gamma}^{(1)}  =  y_{j \neq \delta}^{(2)}  =  0$.
Assuming such profile of nodes} $((\alpha,\beta),(\gamma,\delta))$, we first find 
$((x_{\alpha}^{(1)}, x_{\beta}^{(2)}),( y_{\gamma}^{(1)}, y_{\delta}^{(2)}))$, or equivalently, 
the optimal $((k_g^{(1)}, k_g^{(2)}) , (k_b^{(1)} , k_b^{(2)}))$ corresponding to such a profile.
By incorporating $((\alpha,\beta),(\gamma,\delta))$, Expression (\ref{eqn:raw2phase2camp}) for $\sum_{i \in V} z_i^{(2)}$ simplifies to

\begin{small}
\vspace{-3mm}
\begin{align}
\nonumber
\sum_{i \in V} \sum_{j \in V} & c_i b_{ji} 
+  k_g^{(2)} \frac{\theta_\beta}{2} \Big( \sum_{i \in V}  c_i b_{\beta i} + r_\beta\Big)
 +  k_b^{(2)} \frac{\theta_\delta}{2} \Big( \sum_{i \in V} c_i b_{\delta i} -r_\delta\Big)
\\&
\nonumber
+ k_g^{(1)}  \frac{\theta_\alpha}{2} (1+c_\alpha) \Big( s_{\alpha} 
+ k_g^{(2)} \frac{\theta_\beta}{2} b_{\beta \alpha}
+ k_b^{(2)} \frac{\theta_\delta}{2} b_{\delta \alpha} \Big)
\\&\;
- k_b^{(1)} \frac{\theta_\gamma}{2} (1-c_\gamma) \Big( s_\gamma 
+ k_g^{(2)} \frac{\theta_\beta}{2} b_{\beta \gamma}
+ k_b^{(2)} \frac{\theta_\delta}{2} b_{\delta \gamma} \Big)
\label{eqn:twophase_twocamp_new}
\end{align}
\vspace{-2mm}
\end{small}

First, we consider the case when $k_g^{(1)} + k_g^{(2)} = k_g$ and $k_b^{(1)} + k_b^{(2)} = k_b$.
Now, for a given profile of nodes $((\alpha,\beta),(\gamma,\delta))$, we will find the optimal values of $k_g^{(1)},k_g^{(2)},k_b^{(1)},k_b^{(2)}$. In this case, we have $k_g^{(2)} = k_g - k_g^{(1)}$ and $k_b^{(2)} = k_b - k_b^{(1)}$. 
Substituting this in  (\ref{eqn:twophase_twocamp_new}) and 
equating $\frac{\partial \sum_{i \in V} z_i^{(2)} }{\partial k_g^{(1)}} = 0$, we get

\begin{small}
\vspace{-3mm}
\begin{align}
\nonumber
  - \frac{\theta_\beta}{2} \Big( \sum_{i \in V}  c_i b_{\beta i} + r_\beta\Big)
  + \frac{\theta_\alpha}{2} (1+c_\alpha) \Big( s_{\alpha} + (k_g - 2 k_g^{(1)} ) \frac{\theta_\beta}{2} b_{\beta \alpha} + (k_b - k_b^{(1)} ) \frac{\theta_\delta}{2} b_{\delta \alpha} \Big) 
  \\
 + k_b^{(1)} \frac{\theta_\gamma}{2} (1-c_\gamma) \frac{\theta_\beta}{2} b_{\beta \alpha} = 0
 \nonumber
\end{align}
\vspace{-2mm}
\end{small}

Let $A = \theta_\gamma \theta_\delta (1-c_\gamma) b_{\delta \gamma}$, $D = \theta_\alpha \theta_\beta  (1   +   c_\alpha)b_{\beta \alpha}$, and \\$B = \frac{1}{2} \left( \theta_\alpha \theta_\delta (1+c_\alpha)b_{\delta \alpha} - \theta_\gamma \theta_\beta (1-c_\gamma)b_{\beta \gamma} \right)$. So, the above  can be written as

\begin{small}
\vspace{-3mm}
\begin{align}
D k_g^{(1)} + B k_b^{(1)}
=  - \theta_\beta \Big( \sum_{i \in V}  c_i b_{\beta i} + r_\beta\Big) + \theta_\alpha (1   +   c_\alpha) \Big( s_{\alpha} + k_g  \frac{\theta_\beta}{2} b_{\beta \alpha} + k_b   \frac{\theta_\delta}{2} b_{\delta \alpha} \Big) 
\label{eqn:dv_by_dkg1}
\end{align}
\vspace{-2mm}
\end{small}

Similarly, equating $\frac{\partial \sum_{i \in V} z_i^{(2)} }{\partial k_b^{(1)}} = 0$, we get

\begin{small}
\vspace{-3mm}
\begin{align}
B k_g^{(1)} - A k_b^{(1)}
=  - \theta_\delta \Big( \sum_{i \in V}  c_i b_{\delta i} - r_\delta\Big) - \theta_\gamma (1   -   c_\gamma) \Big( s_{\gamma} + k_g  \frac{\theta_\beta}{2} b_{\beta \gamma} + k_b   \frac{\theta_\delta}{2} b_{\delta \gamma} \Big) 
\label{eqn:dv_by_dkb1}
\end{align}
\vspace{-2mm}
\end{small}

Solving Equations (\ref{eqn:dv_by_dkg1}) and (\ref{eqn:dv_by_dkb1}) simultaneously, we obtain the candidate   $k_g^{(1)}$ for profile $((\alpha,\beta),(\gamma,\delta))$  to be

\vspace{-4mm}
\begin{small}
\begin{align*}
\hspace{-3mm}
\nonumber
\frac{1}{ B^2   +   D A}
\Big[
s_\alpha \theta_\alpha (1   +   c_\alpha)A 
  -   r_\beta \theta_\beta A
  -   s_\gamma \theta_\gamma (1   -   c_\gamma)B
   +   r_\delta \theta_\delta B
\\
\nonumber
+ k_g\Big(\frac{\theta_\alpha\theta_\beta}{2} (1+c_\alpha)b_{\beta\alpha}A - \frac{\theta_\gamma\theta_\beta}{2} (1-c_\gamma)b_{\beta\gamma}B\Big)
\\
\nonumber
+ k_b \Big(\frac{\theta_\alpha\theta_\delta}{2} (1+c_\alpha)b_{\delta\alpha}A - \frac{\theta_\gamma\theta_\delta}{2} (1-c_\gamma)b_{\delta\gamma}B\Big)
\\
-\theta_\beta A \sum_{i \in V} c_i b_{\beta i} - \theta_\delta B \sum_{i \in V} c_i b_{\delta i}
\Big]
\end{align*}
\vspace{-2mm}
\end{small}

We can similarly obtain the candidate $k_b^{(1)}$  for profile $((\alpha,\beta),(\gamma,\delta))$  to be

\vspace{-4mm}
\begin{small}
\begin{align*}
\hspace{-3mm}
\nonumber
\frac{1}{ B^2   +   D A}
\Big[
s_\alpha \theta_\alpha (1   +   c_\alpha)B 
  -   r_\beta \theta_\beta B
  +   s_\gamma \theta_\gamma (1   -   c_\gamma)D
   -   r_\delta \theta_\delta D
\\
\nonumber
+ k_g\Big(\frac{\theta_\alpha\theta_\beta}{2} (1+c_\alpha)b_{\beta\alpha}B + \frac{\theta_\gamma\theta_\beta}{2} (1-c_\gamma)b_{\beta\gamma}D \Big)
\\
\nonumber
+ k_b \Big(\frac{\theta_\alpha\theta_\delta}{2} (1+c_\alpha)b_{\delta\alpha}B + \frac{\theta_\gamma\theta_\delta}{2} (1-c_\gamma)b_{\delta\gamma}D \Big)
\\
-\theta_\beta B \sum_{i \in V} c_i b_{\beta i} + \theta_\delta D \sum_{i \in V} c_i b_{\delta i}
\Big]
\end{align*}
\vspace{-4mm}
\end{small}

If  second derivative w.r.t. $k_g^{(1)}$, that is,
$- {\theta_\alpha} {\theta_\beta} r_\beta w_{\beta\beta}^0 \Delta_{\beta\alpha}  (1+w_{\alpha \alpha}^0 z_\alpha^0)  < 0$
and that w.r.t. $k_b^{(1)}$, that is,
${\theta_\gamma} {\theta_\delta} r_\delta w_{\delta\delta}^0 \Delta_{\delta\gamma}  (1-w_{\gamma \gamma}^0 z_\gamma^0) > 0$, 
and the obtained solution is such that $k_g^{(1)}  \in  [0,k_g]$ and $k_b^{(1)}  \in  [0,k_b]$, then neither the good camp can change $k_g^{(1)}$ to increase $\sum_{i \in V} z_i^{(2)}$, nor the bad camp can change $k_b^{(1)}$ to decrease $\sum_{i \in V} z_i^{(2)}$.
So we can effectively write $u_g((\mathbf{x^{(1)}},\mathbf{x^{(2)}}),(\mathbf{y^{(1)}},\mathbf{y^{(2)}}))$
as $u_g((\alpha,\beta),(\gamma,\delta))$, where $u_g((\alpha,\beta),(\gamma,\delta))$ is the value of $\sum_{i \in V} z_i^{(2)}$,
which corresponds to the strategy profile where  good camp invests on nodes $(\alpha,\beta)$ with  optimal budget split $(k_g^{(1)},k_g^{(2)})$, and  bad camp invests on nodes $(\gamma,\delta)$ with  optimal budget split $(k_b^{(1)},k_b^{(2)})$.

For general case where the solution $k_g^{(1)},k_b^{(1)}$ obtained above may not satisfy $k_g^{(1)} \in [0,k_g]$ and $k_b^{(1)} \in [0,k_b]$, we make practically reasonable assumptions so as to determine $u_g((\alpha,\beta),(\gamma,\delta))$.
We assume that $w_{ij} \geq 0, \forall (i,j)$ and $w_{ii}^0 \geq 0,\theta_i \geq 0, z_i^0 \in [-1,1], \forall {i \in V}$.
For the fund collection example, 
these assumptions respectively mean that  nodes do not distrust each other, they do not weigh their own biases negatively (which is trivially true in the real world), they attribute a non-negative weightage to being influenced due to campaigning as a whole, and they are not excessively biased towards any particular camp before the campaigning and opinion dynamics begin.
Now, it is easy to show that if $w_{ij} \geq 0, \forall (i,j)$, then $\Delta_{ij} \geq 0, \forall (i,j)$ and $r_i \geq 1, \forall {i \in V}$.
So if we assume  $w_{ij} \geq 0, \forall (i,j)$ and $w_{ii}^0 \geq 0,\theta_i \geq 0, z_i^0 \in [-1,1], \forall {i \in V}$,
we would have that 
$- {\theta_\alpha} {\theta_\beta} r_\beta w_{\beta\beta}^0 \Delta_{\beta\alpha}  (1+w_{\alpha \alpha}^0 z_\alpha^0)  \leq 0$
and ${\theta_\gamma} {\theta_\delta} r_\delta w_{\delta\delta}^0 \Delta_{\delta\gamma}  (1-w_{\gamma \gamma}^0 z_\gamma^0) \geq 0$.
That is, we would have $\sum_{i \in V} z_i^{(2)}$ to be a convex-concave function, which is concave w.r.t. $k_g^{(1)}$ and convex w.r.t. $k_b^{(1)}$.
So in the domain $([0,k_g],[0,k_b])$, we can find a $(k_g^{(1)},k_b^{(1)})$ such that,
neither the good camp can change $k_g^{(1)}$ to increase $\sum_{i \in V} z_i^{(2)}$, nor the bad camp can change $k_b^{(1)}$ to decrease $\sum_{i \in V} z_i^{(2)}$ \cite{boyd2004convex,arrow1958studies}.
So we can assign this value $\sum_{i \in V} z_i^{(2)}$ to $u_g((\alpha,\beta),(\gamma,\delta))$.

Thus using  above technique, we obtain  $u_g((\alpha,\beta),(\gamma,\delta))$ 
{for all profiles of nodes $((\alpha,\beta),(\gamma,\delta))$ when
$k_g^{(1)}  +  k_g^{(2)}  =  k_g$} and $k_b^{(1)}  +  k_b^{(2)}  =  k_b$.
From Proposition~\ref{prop:allornothing2camps}, the only other cases to  consider are $k_b^{(1)}  =  k_b^{(2)}  =  0$ and $k_g^{(1)}  =  k_g^{(2)}  =  0$.
Let the profile $((\alpha,\beta),(0,0))$ correspond to $k_b^{(1)}  =  k_b^{(2)}  =  0$.
Note that when  $k_b^{(1)}  =  k_b^{(2)}  =  0$, it reduces to non-competitive case with only  good camp (Section \ref{sec:dep_onecamp}); the value of $\sum_{i \in V} z_i^{(2)}$ for an $(\alpha,\beta)$ pair can hence be assigned to $u_g((\alpha,\beta),(0,0))$.
Thus we can obtain $u_g((\alpha,\beta),(\gamma,\delta))$ for all profiles of nodes $((\alpha,\beta),(0,0))$.
Similarly, we can obtain $u_g((\alpha,\beta),(\gamma,\delta))$ for all profiles of nodes $((0,0),(\gamma,\delta))$.
And from Equation (\ref{eqn:twophase_twocamp_new}), $u_g((0,0),(0,0)) = \sum_{i \in V} \sum_{j \in V} c_i b_{ji} $.

So we have that the good camp has $(n^2+1)$ possible pure strategies to choose from, namely, $(\alpha,\beta) \in V \times V \cup \{(0,0)\}$.
Similarly, the bad camp has $(n^2+1)$ possible pure strategies to choose from, namely, $(\gamma,\delta) \in V \times V \cup \{(0,0)\}$.
We thus have a two-player zero-sum game, for which the utilities of the players can be computed for each strategy profile $((\alpha,\beta),(\gamma,\delta))$ as explained above.
Though we cannot ensure the existence of a pure strategy Nash equilibrium, the finiteness of the number of strategies ensures the existence of a mixed strategy Nash equilibrium. Further, owing to it being a two-player zero-sum game, the Nash equilibrium can be found efficiently by solving a linear program
\cite{osborne2004introduction}.

Summarizing, 
under practically reasonable assumptions
($w_{ij} \geq 0, \forall (i,j)$ and $w_{ii}^0 \geq 0,\theta_i \geq 0, z_i^0 \in [-1,1], \forall {i \in V}$),
we transformed the problem into a two-player zero-sum game with each player having $(n^2+1)$ pure strategies,
and showed how the players' utilities can be computed for each strategy profile.
We
thus deduced the existence of Nash equilibrium and that it can be found efficiently
using linear programming.

\vspace{-3mm}
\section{Simulations and Results
\vspace{-1mm}
}
\label{sec:ODSNmultiphase_sim}

For determining implications of our analytical results on real-world network datasets,
we conducted simulations on 
 NetHEPT dataset (an academic collaboration network obtained from co-authorships in the ``High Energy Physics - Theory'' papers published on the e-print arXiv from 1991 to 2003)
consisting of 15,233 nodes and 31,376 edges.
It is widely used for experimental justifications in the literature on opinion adoption
\cite{kempe2003maximizing,chen2009efficient,chen2010scalable}.
For the purpose of graphical illustration as well as running the computationally intensive algorithm for determining Nash equilibrium for the case of competing camps,
we use the popular Zachary's Karate club dataset consisting of 34 nodes and 78 edges \cite{zachary1977information}.
In all of our simulations, 
we assume the value of  $w_{ii}^0$ to be the same for all nodes, in order to systematically study the effect of this value.

While our model of opinion dynamics (which is an extension of Friedkin-Johnsen model) is based on an iterative process using its update rule, we have arrived at closed form expressions for computing $\mathbf{z^{(1)}}$ and $\mathbf{z^{(2)}}$, and hence $\sum_{i\in V} z_i^{(1)}$ and $\sum_{i\in V} z_i^{(2)}$.
So instead of running the opinion dynamics until convergence for arriving at their converged values, we compute the values of the corresponding closed form expressions.
Our approaches for finding the  nodes to be invested on in the two phases, in both the non-competitive and competitive cases, involve iterating over node tuples for computing the corresponding budget splits and the value of $\sum_{i\in V} z_i^{(2)}$. Hence, our algorithms have adequate scope for parallel computing.
We implemented it using MATLAB and ran on an 8-threaded Intel Core i7 2.8 GHz processor. 
For given values of $v_i^0$ and $w_{ii}^0$, the running time of our algorithm  for the NetHEPT dataset (non-competitive case) was around 3 hours.
For the Karate dataset, the running time for the non-competitive case was less than 0.1 seconds, while that for the competitive case was around 1 minute.

\vspace{-1mm}
\subsection{Simulation Results: The Non-Competitive (Single Camp) Case}

For the single camp case, we consider $z_i^0 = 0, \forall {i \in V}$ (unless specified otherwise) to start with a neutral network. This would help us to  reliably study the effects of critical parameters $w_{ii}^0,r_i$, and hence also $s_j$.
Note that $z_i^0 = 0, \forall {i \in V}$ would mean that $w_{ig}^{(1)} = w_{ib}^{(1)}$ (from (\ref{eqn:WonV})), that is, we start with an unbiased population.

\begin{figure}[h!]
\centering
\includegraphics[width=0.6\textwidth]{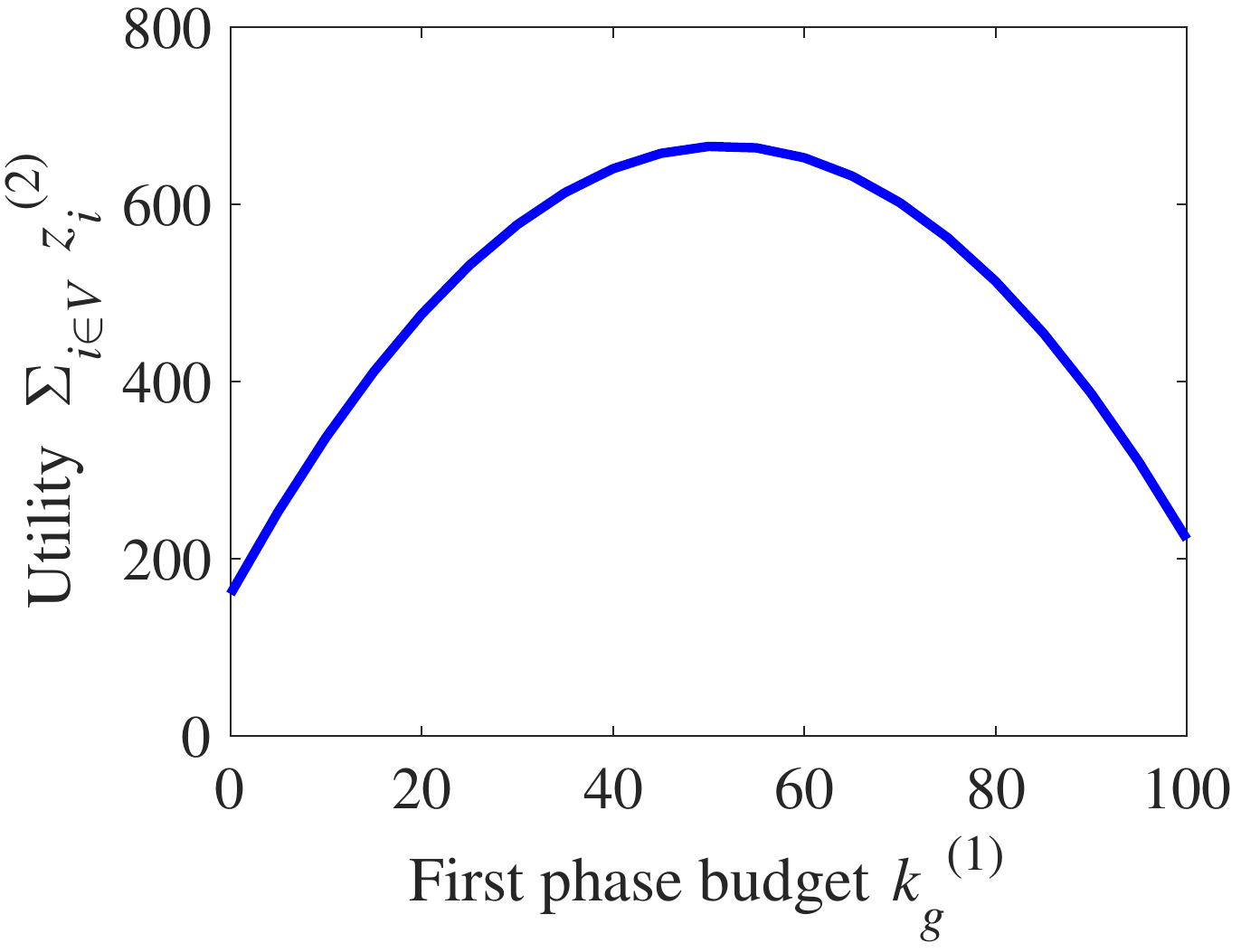}
\vspace{-3mm}
\caption{
The effect of
different budget splits on the good camp's utility (NetHEPT) with $k_g=100$ for $w_{ii}^0 = 0.5, \forall i \in V$ ($z_i^0=0, \forall i \in V$)}
\label{fig:budgetsplits}
\vspace{-2mm}
\end{figure}

Figure \ref{fig:budgetsplits}
shows how the good camp's utility $\sum_{i\in V} z_i^{(2)}$  changes as it increases the budget allotment for the first phase (while decreasing the budget allotment for the second phase) for NetHEPT dataset with $k_g=100$ for $w_{ii}^0 = 0.5, \forall i \in V$.
A natural tradeoff is observed in the plot since 
investing heavily in either phase  results in a low utility; the optimal budget allocation for the first phase is thus an intermediate value. The tradeoff arises since a lower investment in the first phase results in worse initial biases for the second phase in the network, thus resulting a poorer opinion values to start with in the second phase, and also lower effectiveness of its investment in the second phase. On the other hand, a higher investment in the first phase spares a lower available budget for  second phase, resulting in the camp being unable to fully harness the influenced biases. 

\vspace{-2mm}
\subsubsection{The effect of $w_{ii}^0$}

\begin{figure}[h!]
\centering
\includegraphics[width=0.6\textwidth]{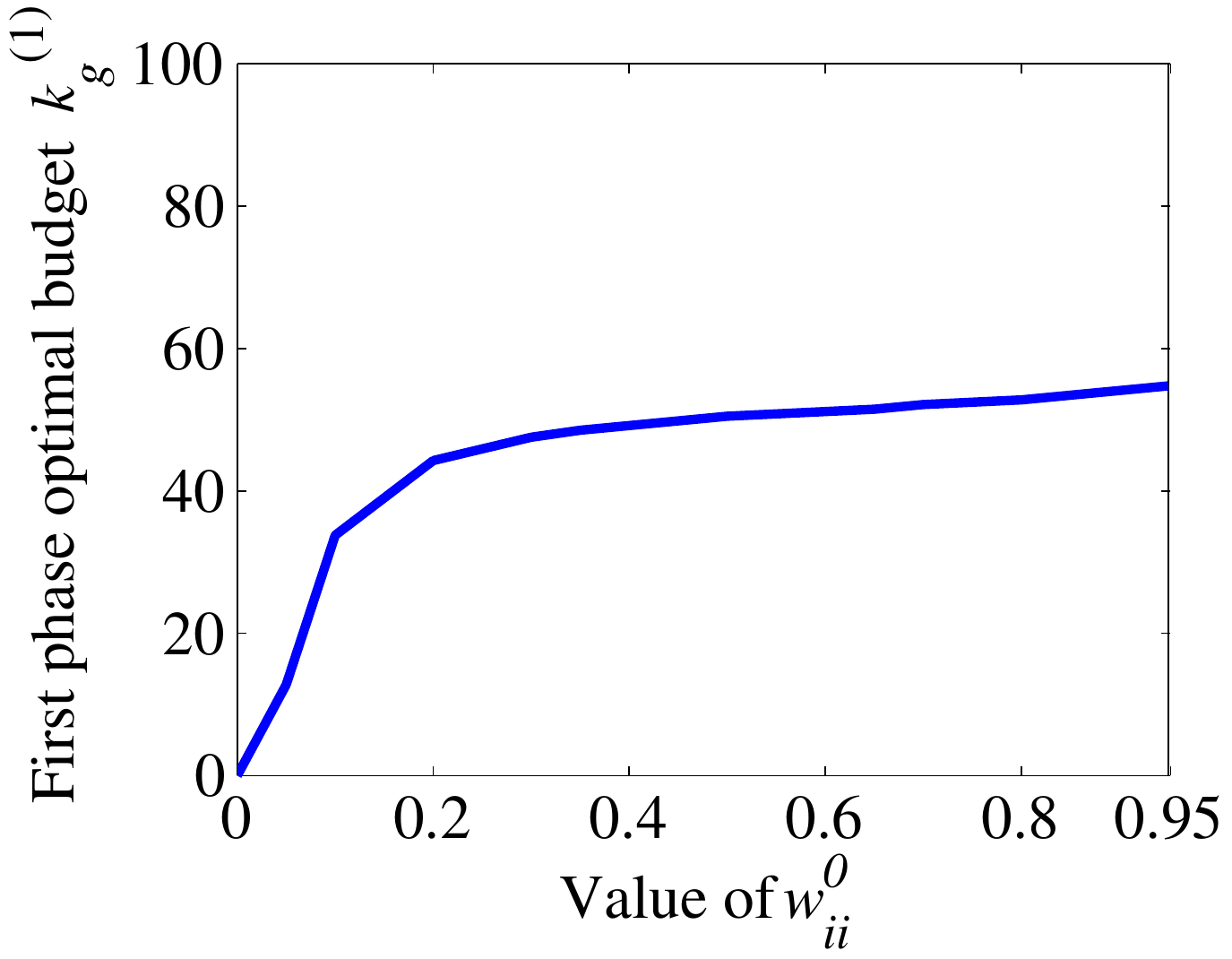}
\vspace{-3mm}
\caption{
The effect of
 $w_{ii}^0$ (NetHEPT) with $k_g=100$ ($z_i^0=0, \forall i \in V$)}
\label{fig:dependency}
\vspace{-2mm}
\end{figure}

Figure \ref{fig:dependency}
presents the optimal budget that should be allotted for the first phase as a function of $w_{ii}^0$.
In our simulations, the optimal values obtained are such that $k_g^{(2)} = k_g - k_g^{(1)}$.
For low values of $w_{ii}^0$,
the optimal strategy is to invest almost entirely in the second phase. This is because the effect of the first phase diminishes in the second phase when $w_{ii}^0$ is low. 
Remark~\ref{rem:singlecamp} states that a high $s_j$ value (influencing power of $j$ looking one phase ahead) would attract high investment in the first phase.
The value $s_j  =  \sum_{i \in V} r_i w_{ii}^0 \Delta_{ij}$ would be significant only if $j$ influences nodes $i$ with significant values of $w_{ii}^0$. 
With low $w_{ii}^0$, we are less likely to have node $j$ with high value of $s_j$ since it requires it to be influential towards significant number of nodes $i$ with significant values of $w_{ii}^0$.
Hence, allotting a significant budget for the first phase is advantageous only if we have nodes with  significant value of $w_{ii}^0$.

\subsubsection{The effect of $z_i^0$}

To supplement the above observation,
we  observed the effects of $z_i^0\neq 0$ (initially biased network).
These plots did not differ noticeably 
from the one in Figure~\ref{fig:dependency}, 
even for higher magnitudes of $z_i^0$.
There are subtle differences, however.
If $z_i^0$'s are  positive, the good camp invests less in the first phase. This is  because $z_i^0$'s already give a  head start for a healthy value at the end of the first phase (which is bias for the second phase), and so the budget could rather be invested in the second phase.
On the other hand, if $z_i^0$'s are  negative, it invests more in the first phase so as to nullify the initial disadvantage.

\subsubsection{The effect of camp being myopic}

\begin{figure}[h!]
\centering
\begin{tabular}{c}
\hspace{-5mm}
\includegraphics[width=0.6\textwidth]{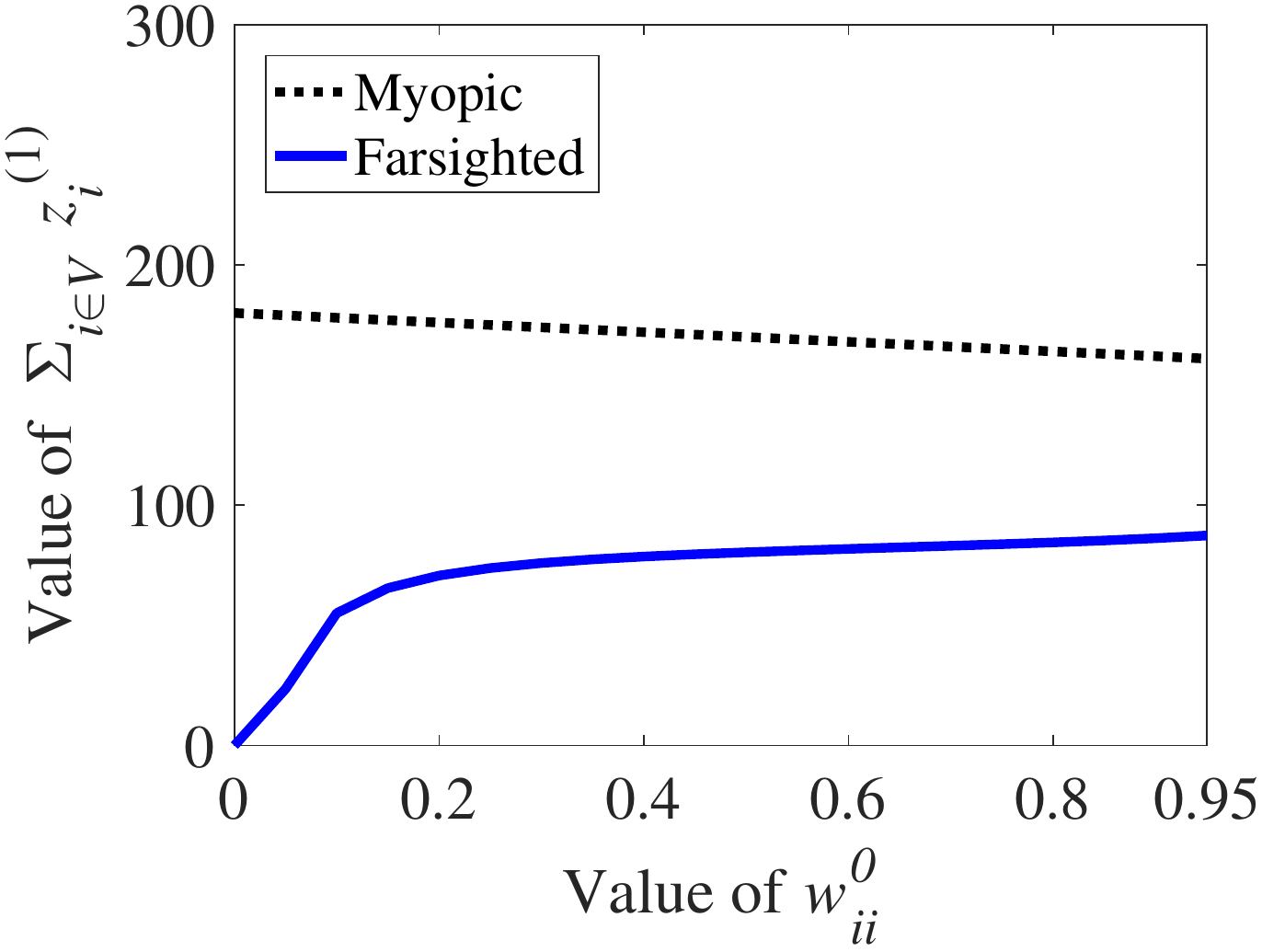}
\\
(a) at the end of the first phase
\vspace{5mm}
\\
\includegraphics[width=0.6\textwidth]{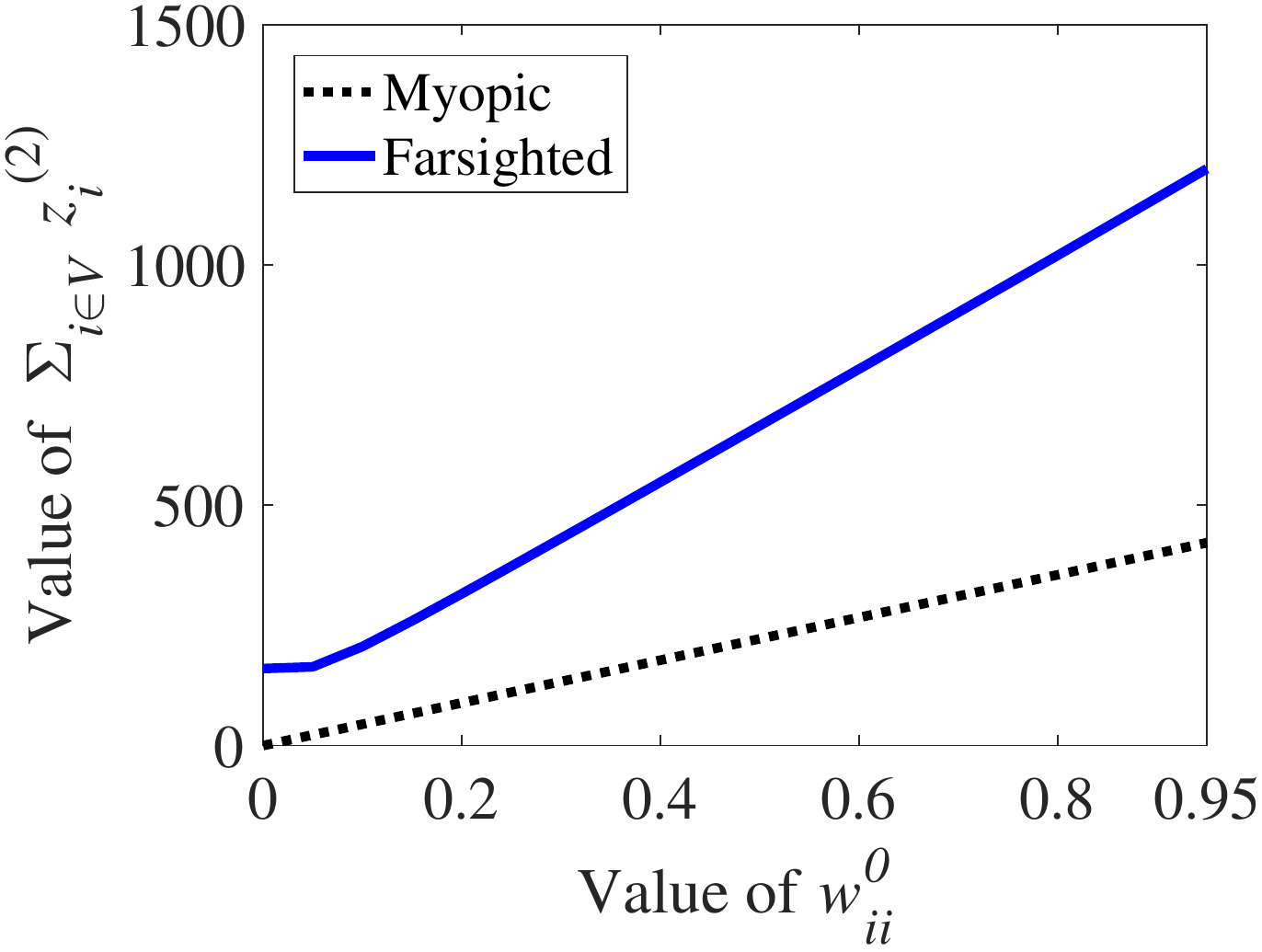}
\\
(b) at the end of the second phase
\end{tabular}

\caption{
The effect of using myopic investment strategy
for different values of $w_{ii}^0$ (NetHEPT) with $k_g=100$ ($z_i^0=0, \forall i \in V$)}
\label{fig:myopic}
\vspace{-2mm}
\end{figure}

We generally consider that the camp is
{\em farsighted\/}, that is, it computes its strategy considering that there would be a second phase, that is, it considers the objective function $\sum_{i\in V} z_i^{(2)}$.
A camp can be called {\em myopic\/} if it computes its strategy greedily by considering the near-sighted objective function $\sum_{i\in V} z_i^{(1)}$.
In other words, when the camp is myopic,
it perceives its utility as
$\sum_{i\in V} z_i^{(1)}$ and devises its strategy to invest greedily in the first phase, even
though its actual utility is $\sum_{i\in V} z_i^{(2)}$.

Figure \ref{fig:myopic} shows 
the effect of the good camp using myopic strategy (investing the entire budget in the first phase)
over  the wide range of $w_{ii}^0$ values for NetHEPT dataset with $k_g=100$.
It is clear that the myopic strategy would result in a higher value of $\sum_{i\in V} z_i^{(1)}$ than that achieved using the farsighted strategy, since the myopic one invests the entire budget in the first phase, while the farsighted one invests as per Figure~\ref{fig:dependency}.
This can be seen from Figure~\ref{fig:myopic}(a).
In Figure~\ref{fig:myopic}(b),
when $w_{ii}^0=0$, the myopic strategy results in zero utility 
as it invests its entire budget in the first phase (which plays no role when $w_{ii}^0=0$), while the  farsighted strategy suggests the camp to invest its entire budget in the second phase (which is why it results in $\sum_{i\in V} z_i^{(1)} = 0$ when $w_{ii}^0=0$ in Figure~\ref{fig:myopic}(a)).
Figure~\ref{fig:myopic}(b) also suggests that the loss incurred by playing myopic strategy could be considerably high for high values of $w_{ii}^0$. This emphasizes  that, though high $w_{ii}^0$'s are suitable for high investments in the first phase so as to influence the biases for the second phase (as seen in Figure~\ref{fig:dependency}), it is important to spare a certain fraction of the budget so as to be invested in the second phase in order to harness the influenced biases.

\vspace{-2mm}
\subsubsection{Phasewise progression of opinion values}
\label{sec:phasewise}

In order to illustrate the phasewise progression of opinion values of nodes, we use the small-sized Zachary's Karate club dataset  for  visualization.
Figure~\ref{fig:Ri_Si} 
shows the computed values of $s_i$ and $r_i$ for the nodes in the dataset. 
The size and color saturation of a node $i$ represent the value of the corresponding parameter (bigger size and higher saturation implies higher value).
For this dataset, the investment was made on a single node common to both the phases. This node  visibly stands out in Figures \ref{fig:phasewise_0.5} and \ref{fig:phasewise_0.9} with its size and color saturation;
we refer to this as the prime node in our discussion.

\begin{figure}[h!]
\vspace{2mm}
\centering
\begin{tabular}{cc}
\includegraphics[width=0.5\textwidth]{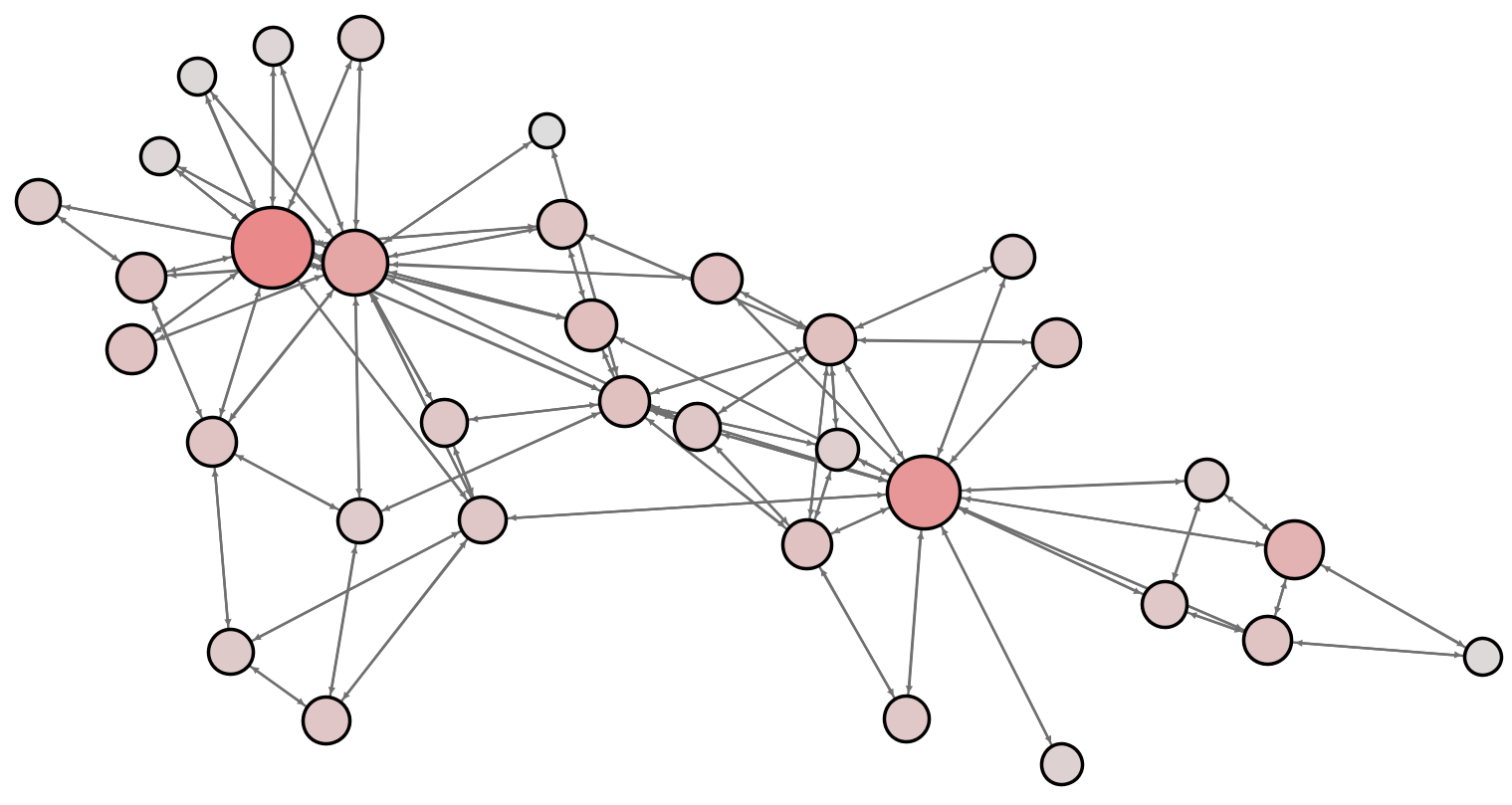}
&
\includegraphics[width=0.5\textwidth]{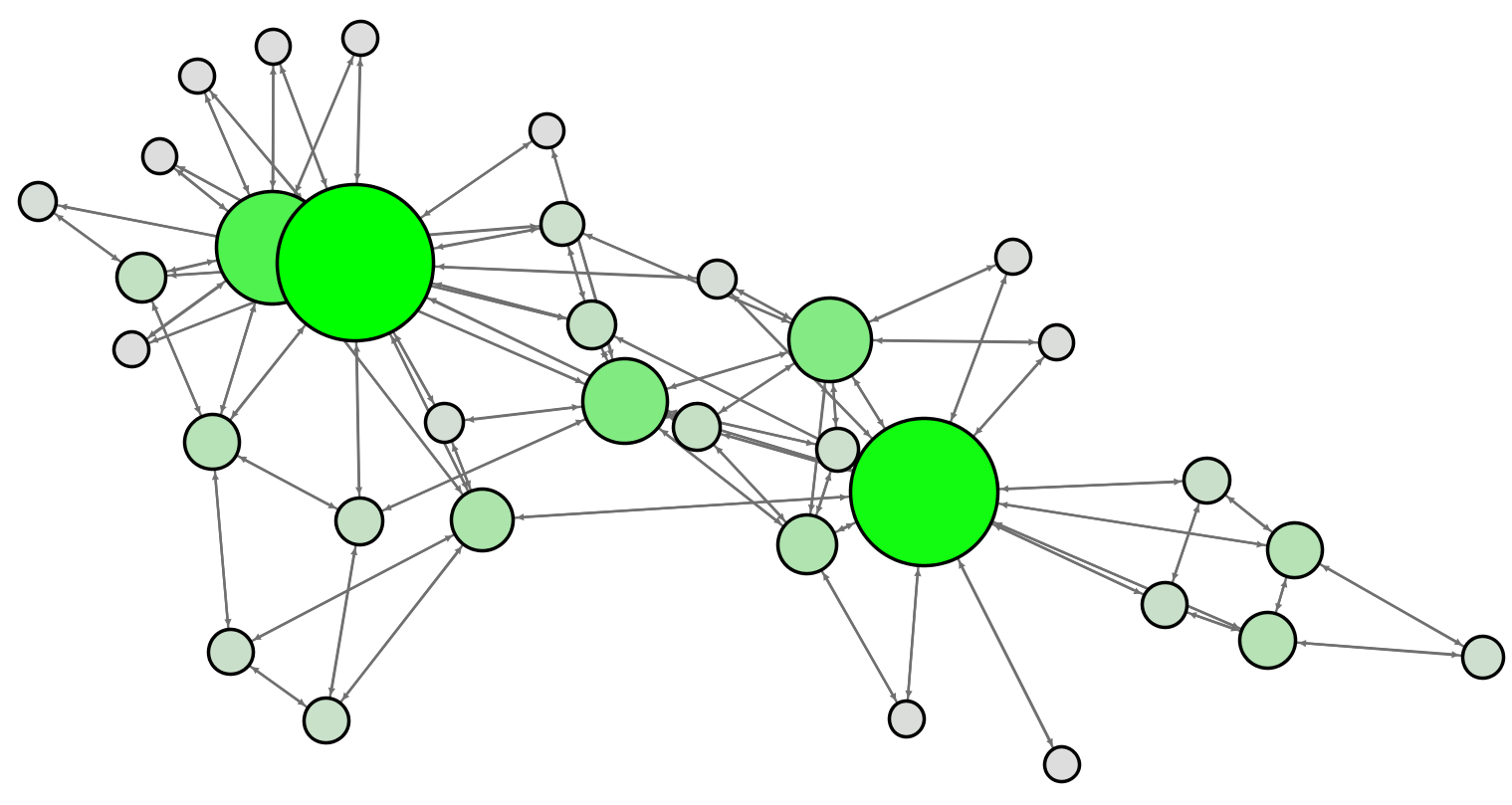}
\\
(a) values of $s_i$
&
(b) values of $r_i$
\end{tabular}
\vspace{-5mm}
\caption{
Nodes' relative values of $s_i$ and $r_i$}
\label{fig:Ri_Si}
\end{figure}

\begin{figure}[h!]
\centering
\begin{tabular}{cc}
\includegraphics[width=0.5\textwidth]{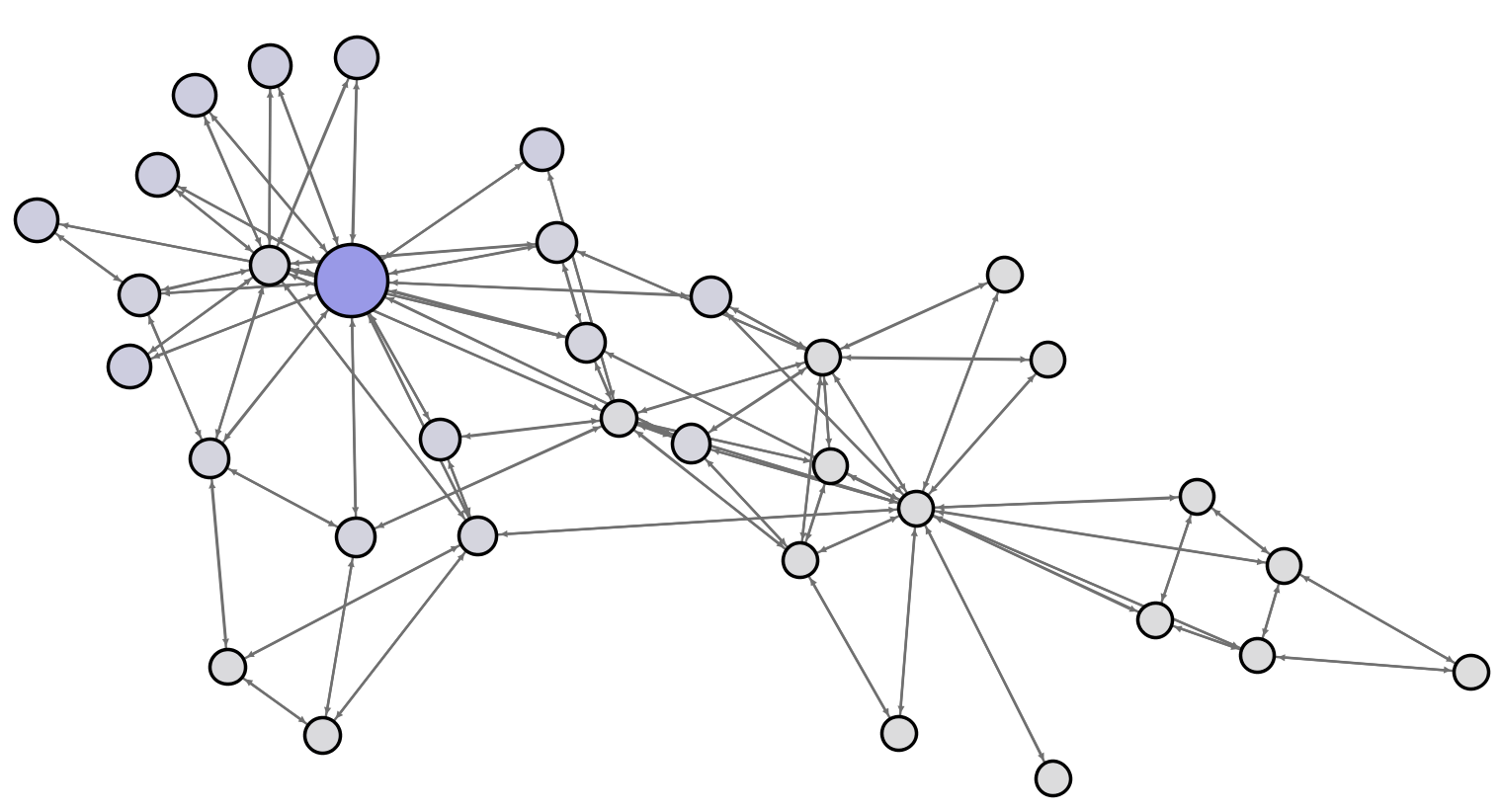}
&
\includegraphics[width=0.5\textwidth]{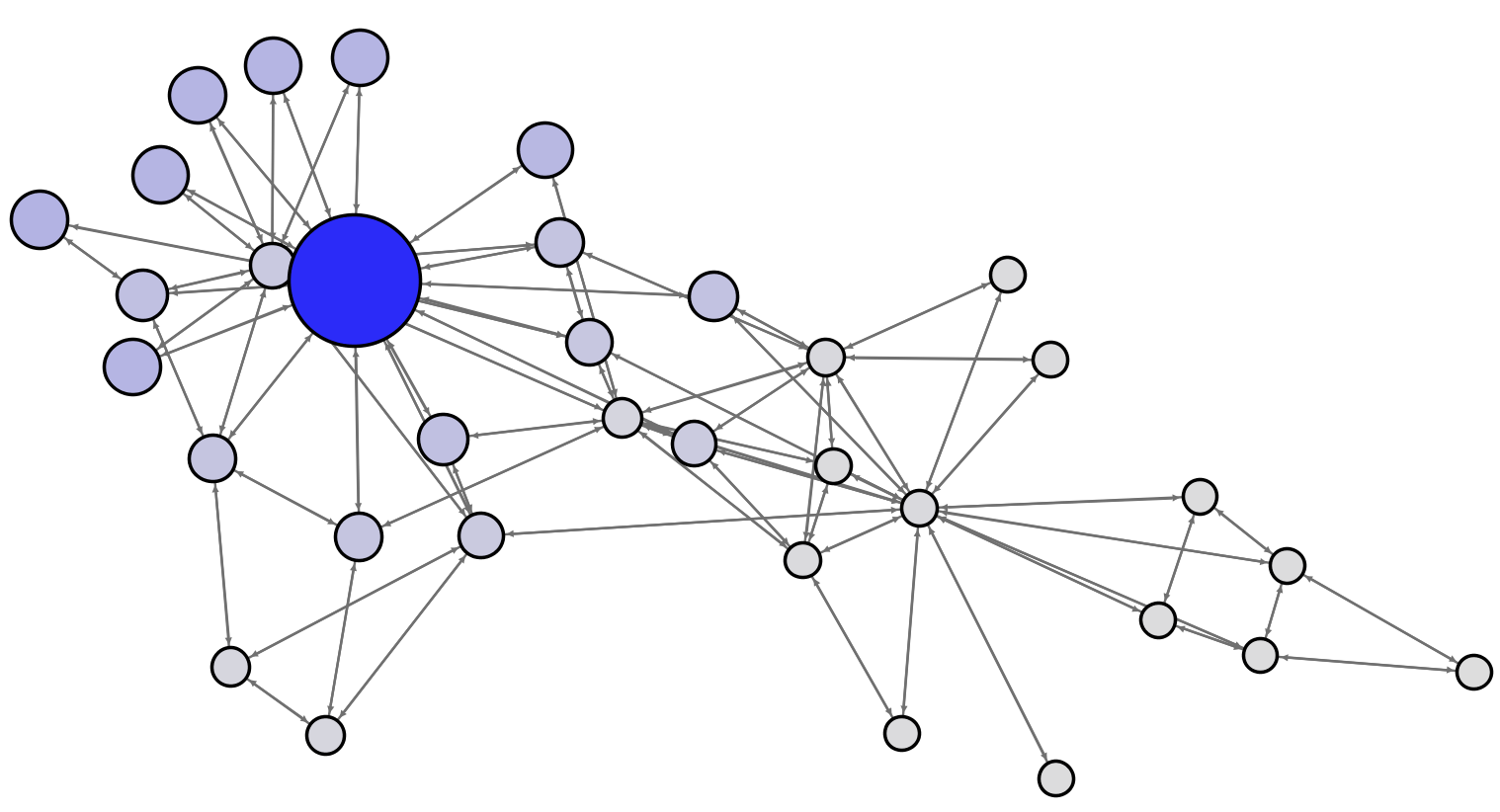}
\\
(a) at the end of first phase
&
(b) at the end of second phase
\end{tabular}
\vspace{-4mm}
\caption{
Illustration of phasewise progression of opinion values when $w_{ii}^0= 0.5$ (Karate) with $k_g=5$ ($z_i^0=0, \forall i \in V$)}
\label{fig:phasewise_0.5}
\end{figure}

\begin{figure}[h!]
\centering
\begin{tabular}{cc}
\includegraphics[width=0.5\textwidth]{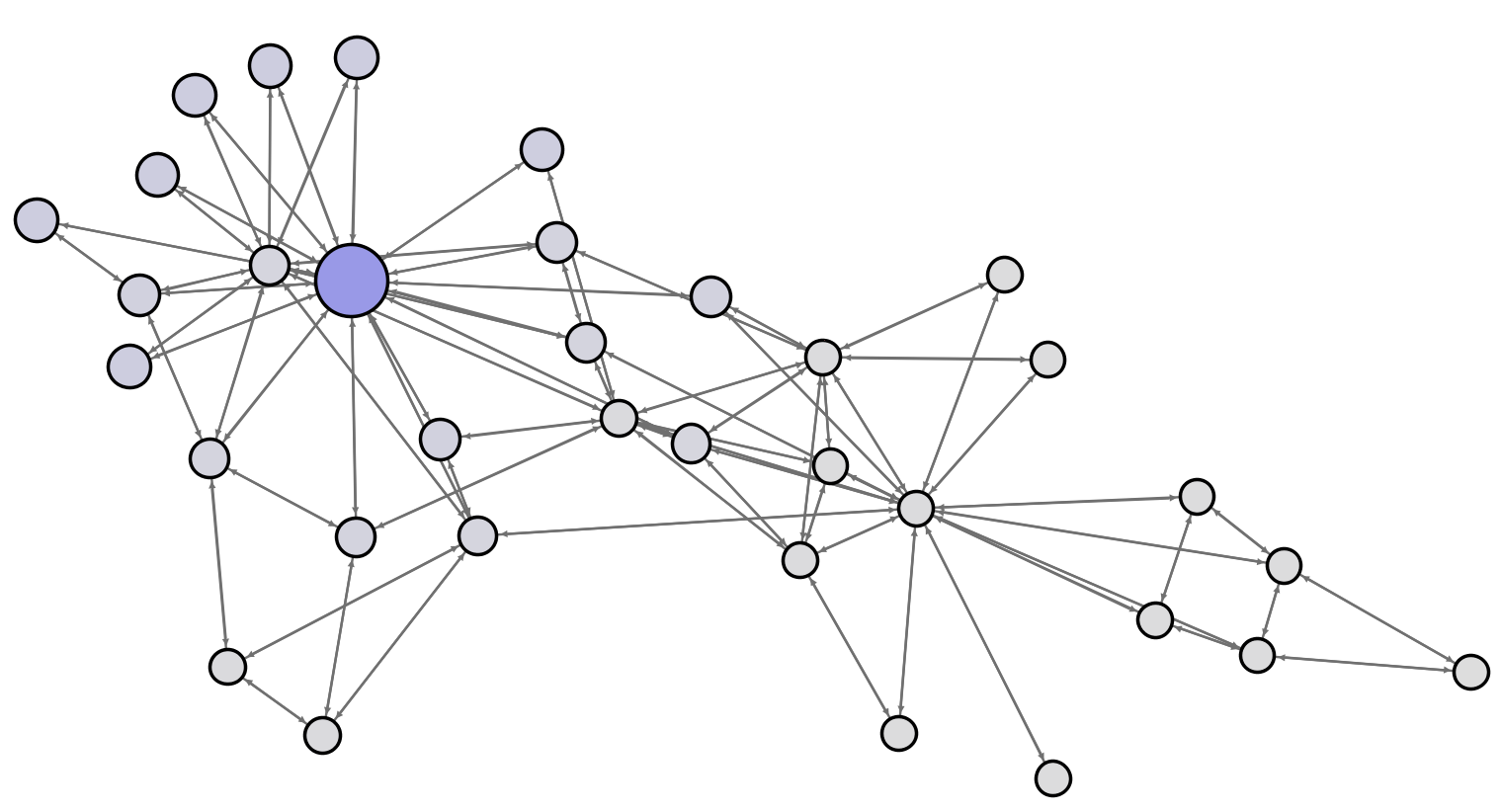}
&
\includegraphics[width=0.5\textwidth]{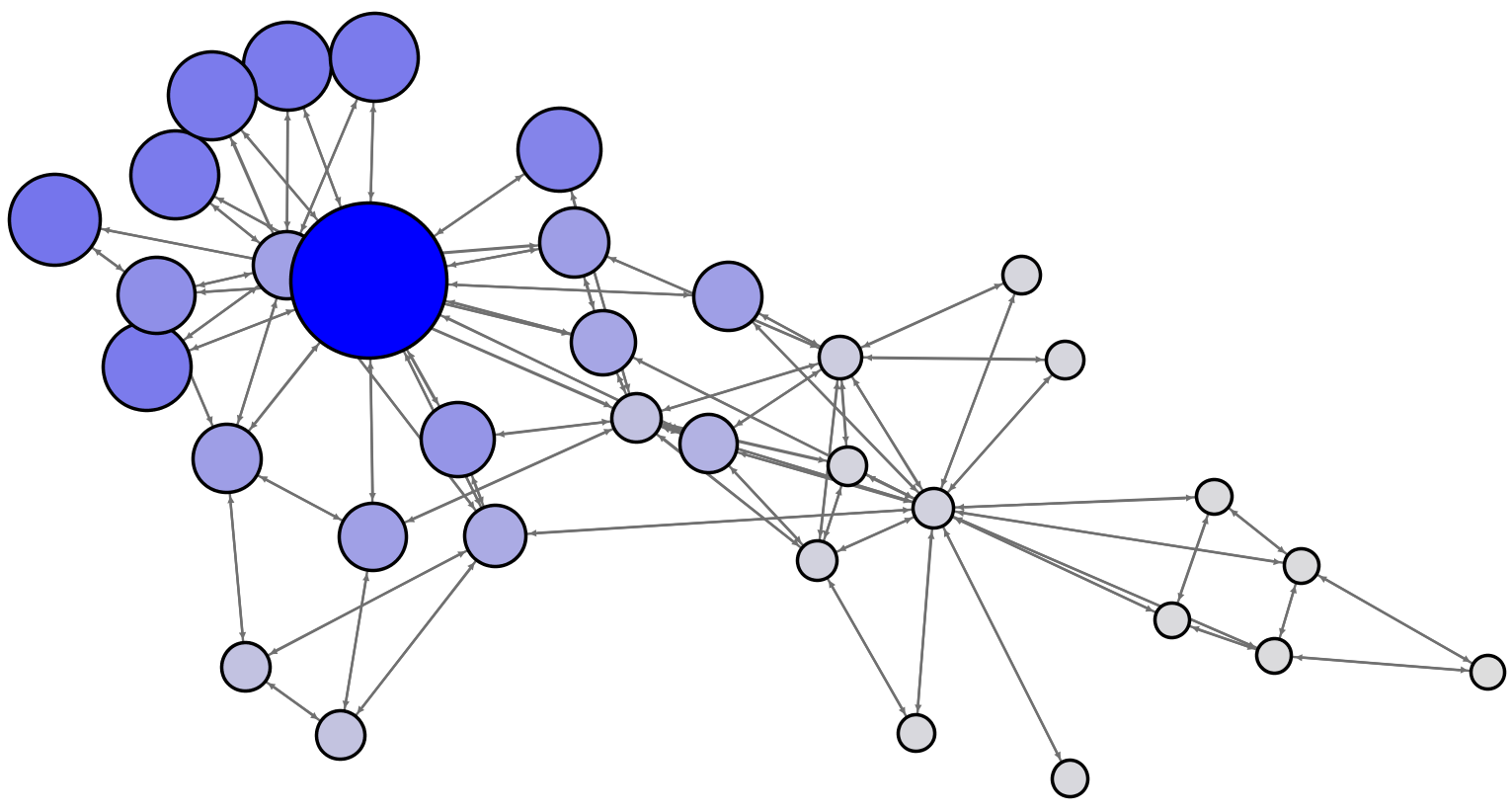} 
\\
(a) at the end of first phase
&
(b) at the end of second phase
\end{tabular}
\vspace{-4mm}
\caption{
Illustration of phasewise progression of opinion values when $w_{ii}^0= 0.9$ (Karate) with $k_g=5$ ($z_i^0=0, \forall i \in V$)}
\label{fig:phasewise_0.9}
\end{figure}

In Figures \ref{fig:phasewise_0.5} and \ref{fig:phasewise_0.9},
the size and color saturation of a node $i$ represent its opinion value  (bigger size and higher saturation implies higher opinion value).
It can be seen that a higher value of $w_{ii}^0$ results in a significant change in opinion values in the second phase.
This is owing to the fact that the good camp's investment is more effective in the second phase when nodes attribute higher weightage to their initial biases in the second phase, or equivalently, their opinions at the end of the first phase (Equation~(\ref{eqn:WonV})) (assuming positive opinion values which is the case here).

\subsubsection{Other Possible Approaches}

Since our approach requires $n^2$ iterations so as to search over all pairs of nodes, we explore other possible approaches which could be used for determining a way of splitting the total available budget across the two phases as well as the nodes to be invested on in the two phases.

\begin{figure}[h!]
\centering
\includegraphics[width=0.6\textwidth]{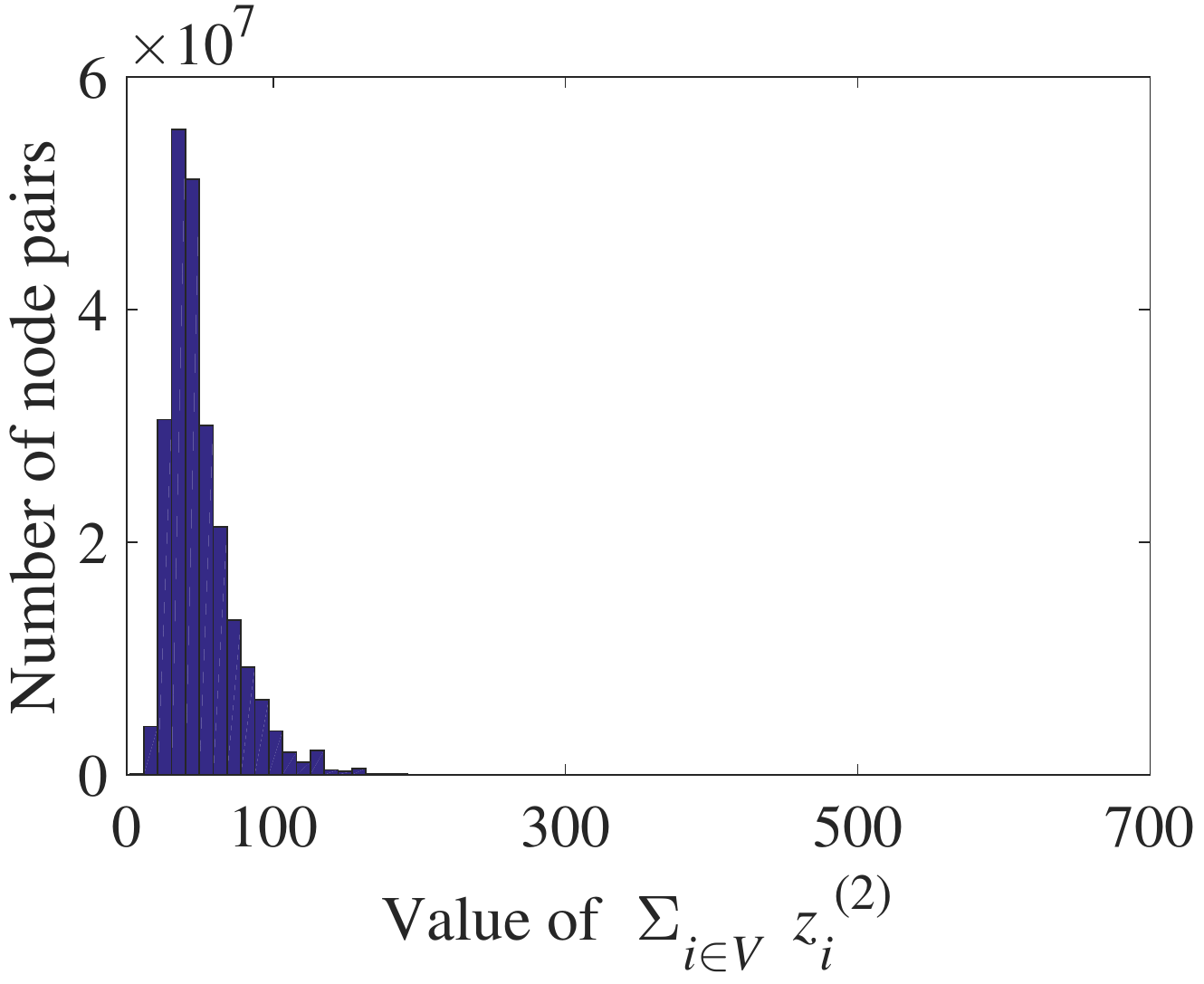}
\vspace{-3mm}
\caption{
Histogram showing the number of node pairs resulting in different values of $\sum_{i\in V} z_i^{(2)}$ for
 $w_{ii}^0 = 0.5$ for NetHEPT dataset with $k_g=100$ ($z_i^0=0, \forall i \in V$)}
\label{fig:V2_hist_pairs}
\vspace{-2mm}
\end{figure}

\begin{figure}[h!]
\centering
\includegraphics[width=0.6\textwidth]{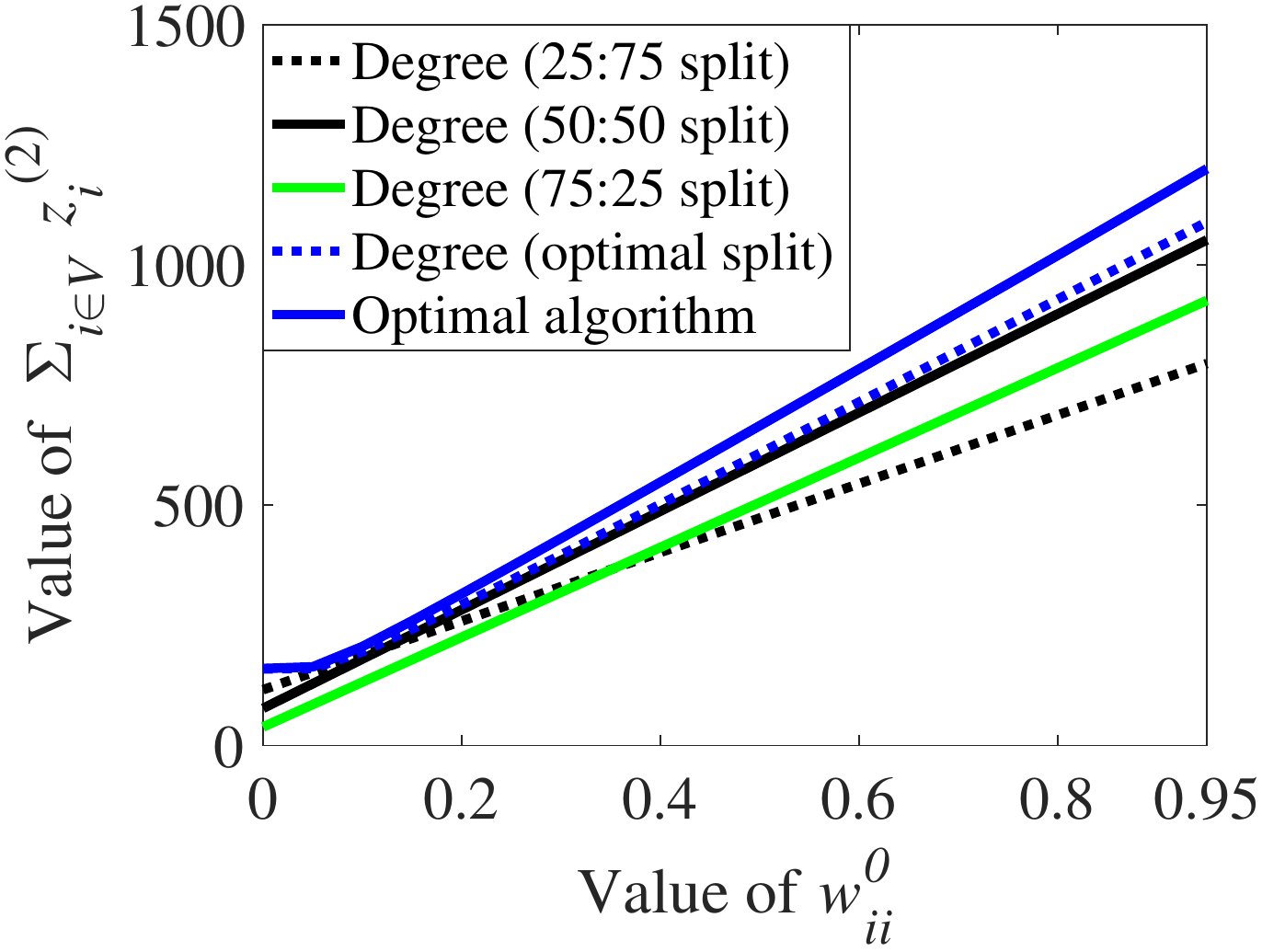}
\vspace{-3mm}
\caption{
Performance of high degree heuristic with different budget splits versus the optimal algorithm
 (NetHEPT) with $k_g=100$  ($z_i^0=0, \forall i \in V$)}
\label{fig:all_algos}
\vspace{-2mm}
\end{figure}

We first study how a random algorithm would perform. In particular, we study  the distribution of the values of $\sum_{i\in V} z_i^{(2)}$ over all pairs of nodes, assuming that the budget split for a given pair is obtained using Expressions (\ref{eqn:optk1_new}) and (\ref{eqn:optk2_new}).
Figure \ref{fig:V2_hist_pairs}
shows the histogram for the NetHEPT dataset with $k_g=100$, considering $w_{ii}^0 = 0.5$ (the histograms for other values of $w_{ii}^0$ were qualitatively similar).
While the value of $\sum_{i\in V} z_i^{(2)}$ corresponding to the optimal node pair is  $666.02$, it can be seen that most pairs resulted in value less than $100$.
Hence, it is clear that a random node selection algorithm would perform quite badly, even with our approach of optimally splitting the budget across the two phases.

However, instead of selecting nodes randomly, certain node centrality measures could be used in order to find the important nodes which can be invested on.
The budget could be split over the two phases, either in a manually chosen ratio, or using Expressions (\ref{eqn:optk1_new}) and (\ref{eqn:optk2_new}) corresponding to the chosen nodes.
Figure~\ref{fig:all_algos} presents the results of the high degree heuristic, where the highest degree node is chosen to be invested on in both the phases. The different ways of splitting the budget are, hence, compared with our optimal approach.
A $25$:$75$ split ($k_g^{(1)}=25, k_g^{(2)}=75$) performs well for low values of $w_{ii}^0$, since as explained earlier, a low $w_{ii}^0$ necessitates a higher investment in the second phase. On the other hand, a $50$:$50$ split performs well for moderate and high values of $w_{ii}^0$, since the optimal budget split for this node pair obtained using (\ref{eqn:optk1_new}) and (\ref{eqn:optk2_new}) was close to $50$:$50$.
For the same reason, the difference in the performances of the $50$:$50$ split and the optimal split is not very significant.

While the node selection approaches based on node centrality measures are likely to perform well in practice, we explore another greedy approach which explicitly accounts for our objective function (\ref{eqn:twophase_onecamp_1})
while selecting the nodes in the two phases.
Specifically, while choosing the node for the first phase, we ignore any investment that would occur in the second phase. That is, in Expression~(\ref{eqn:twophase_onecamp_1}), we assume $k_g^{(2)}=0$. So 
the node to be chosen for the first phase is the node $\alpha$ with the highest value of $\theta_{\alpha}(1+c_{\alpha}) s_{\alpha}$.
Now assuming this node as the one that would be invested on in the first phase, we iterate over the $n$ nodes to be invested on in the second phase (while determining the budget split as per (\ref{eqn:optk1_new}) and (\ref{eqn:optk2_new})), so that the value of $\sum_{i\in V} z_i^{(2)}$ is maximized.
For the studied NetHEPT dataset, the nodes chosen for the two phases using this greedy approach turn out to be the highest degree node. Hence, this approach performs equal to the high degree heuristic with the optimal split (Figure~\ref{fig:all_algos}) for the NetHEPT dataset.
However, in general, the greedy approach is a promising one, since it partially accounts for the objective function, while reducing the number of iterations from $n^2$ to $2n$.

\subsection{Simulation Results: The Case of Competing Camps}

Though we presented a polynomial time algorithm for determining Nash equilibrium when there are two camps, which is of theoretical interest, it is computationally expensive to run it on larger networks.
Hence, for the purpose of studying the  case of competing camps,
we consider  the  Zachary's Karate club dataset (34 nodes, 78 edges) \cite{zachary1977information}.

Note that
for $z_i^0=0$, the two-phase investment game turns out to be symmetric, since both camps would have the same effectiveness in the first phase, that is, $w_{ig}^{(1)} = w_{ib}^{(1)}, \forall i \in V$.
The other parameters such as $w_{ij}$ and $w_{ii}^0$, and hence $s_i$ and $r_i$, are common to both the camps. 
So, the Nash equilibrium strategy played by the camps is symmetric.
Owing to the game being zero-sum,
 this results in both camps receiving zero utility, that is, $\sum_{i\in V} z_i^{(2)} = 0$.

\subsubsection{The effect of $w_{ii}^0$}

Since $z_i^0=0$ results in trivial observations, we study the effect of  $w_{ii}^0$, while first considering $z_i^0=+0.1$.
In general, we observed that high values of $w_{ii}^0$ resulted in a pure strategy Nash equilibrium in which the camps invested their entire budget on the prime node mentioned in Section \ref{sec:phasewise}.
For low and intermediate values of $w_{ii}^0$, we computed mixed strategy Nash equilibrium, describing the probability with which the good camp would invest on the pair $(\alpha,\beta)$ (node $\alpha$ in the first phase and node $\beta$ in the second phase) and the probability with which the bad camp would invest on the pair $(\gamma,\delta)$, with the corresponding saddle point budget splits as derived in Section \ref{sec:dep_2camps}.
Hence, in order to study the budget allotted to the first phase for a given value of $w_{ii}^0$, we compute the expectation (weighted average where the weights correspond to the probabilities) of the first phase investments corresponding to the aforementioned pairs of nodes.

Formally, in mixed strategy Nash equilibrium, let the good camp play its strategy $(\alpha,\beta)$ with probability $\mathbb{P}_g(\alpha,\beta)$ and the bad camp play $(\gamma,\delta)$ with probability $\mathbb{P}_b(\gamma,\delta)$. Let the corresponding saddle point suggest the good camp to invest $k_g^{(1)} ((\alpha,\beta),(\gamma,\delta))$ on node $\alpha$, and the bad camp to invest $k_b^{(1)} ((\alpha,\beta),(\gamma,\delta))$ on node $\gamma$, in the first phase.
So the expected (weighted average) first phase investment by the good camp is $\sum_{((\alpha,\beta),(\gamma,\delta))} \mathbb{P}_g(\alpha,\beta) \cdot \mathbb{P}_b(\gamma,\delta) \cdot k_g^{(1)} ((\alpha,\beta),(\gamma,\delta))$
and that by the bad camp is $\sum_{((\alpha,\beta),(\gamma,\delta))} \mathbb{P}_g(\alpha,\beta) \cdot \mathbb{P}_b(\gamma,\delta) \cdot k_b^{(1)} ((\alpha,\beta),(\gamma,\delta))$.

For various values of $w_{ii}^0$ in general, the camps invested significantly  on the prime node (mentioned in Section \ref{sec:phasewise}) with high probability. Investments were also made with considerable probability on the node with the highest value of $s_i$ in the first phase and the second highest value of $r_i$ in the second phase, in Figure \ref{fig:Ri_Si}. Certain other nodes with good enough values of $s_i$ and $r_i$ were invested on with low non-zero probabilities.

\begin{figure}[h!]
\centering
\includegraphics[width=0.6\textwidth]{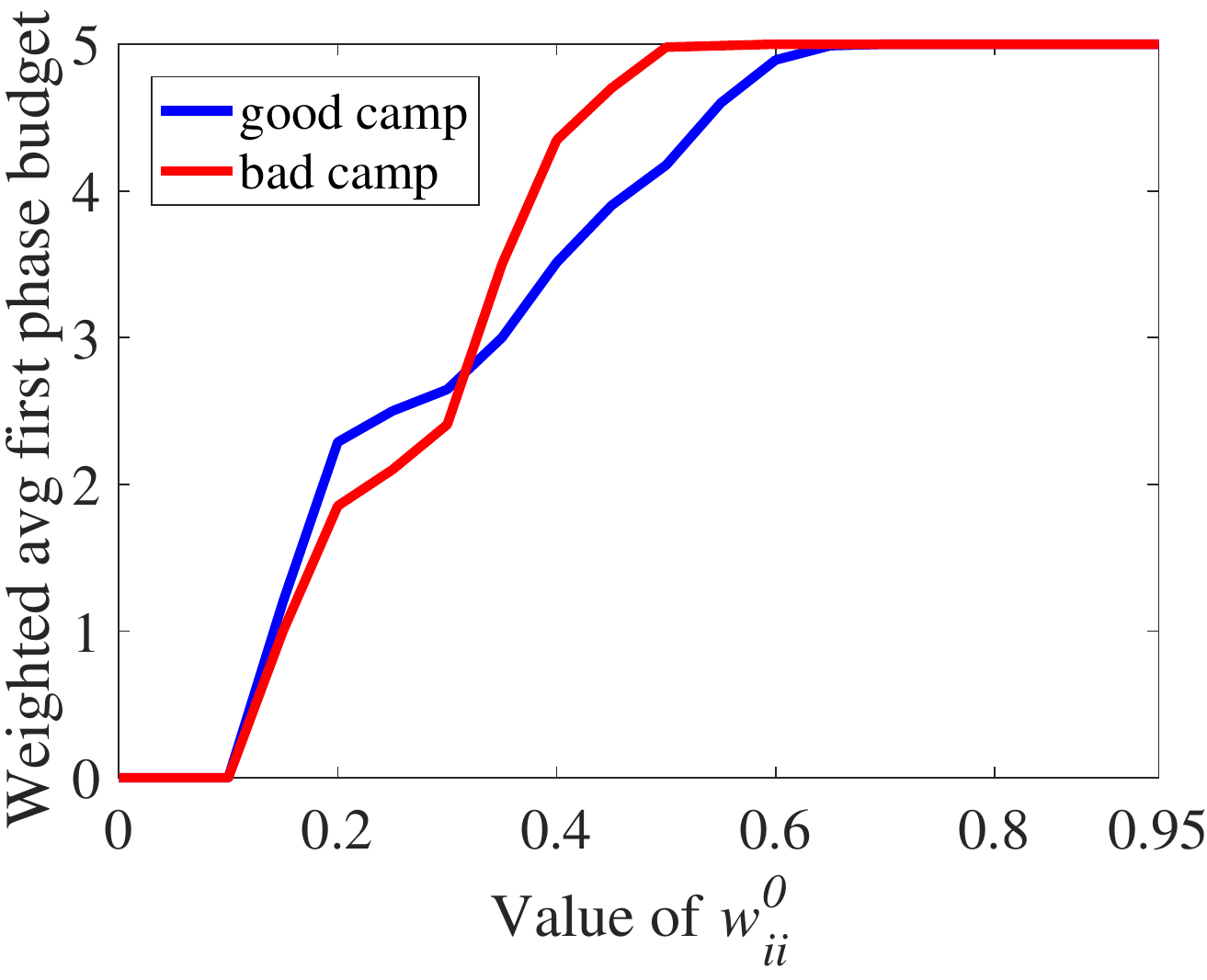}
\vspace{-3mm}
\caption{
The effect of
 $w_{ii}^0$ (Karate) with $k_g=k_b=5$ ($z_i^0=+0.1, \forall i \in V$)}
\label{fig:MSNE_dependency}
\end{figure}

Figure \ref{fig:MSNE_dependency} presents the weighted average (expected) investments in the first phase by the two camps in the mixed strategy Nash equilibrium, corresponding to different values of $w_{ii}^0$.
Similar to the single camp case, the trend is similar in that, the budget allotted to the first phase increases with $w_{ii}^0$. A high $w_{ii}^0$ encourages the camps to influence the initial biases for the second phase, which not only plays a key role in determining the final opinion owing to the high weightage attributed to the biases in the second phase (which are same as the opinions at the end of the first phase), but also enhances the effectiveness of the camps' investments in the second phase.

\subsubsection{The effect of $z_i^0$}

Recall from Section \ref{sec:dep_2camps} that $u_g ((\alpha,\beta),(\gamma,\delta))$ is the good camp's utility when it plays the pure strategy $(\alpha,\beta)$ and the bad camp plays $(\gamma,\delta)$ with the corresponding saddle point budget splits.
Following our above discussion on the probabilities in mixed strategy Nash equilibrium, the value of $\sum_{i\in V} z_i^{(2)}$, quantifying the expected utility of the good camp in mixed strategy Nash equilibrium, is 
$\sum_{((\alpha,\beta),(\gamma,\delta))} \mathbb{P}_g(\alpha,\beta) \cdot \mathbb{P}_b(\gamma,\delta) \cdot u_g ((\alpha,\beta),(\gamma,\delta))$.

\begin{figure}[h!]
\centering
\includegraphics[width=0.6\textwidth]{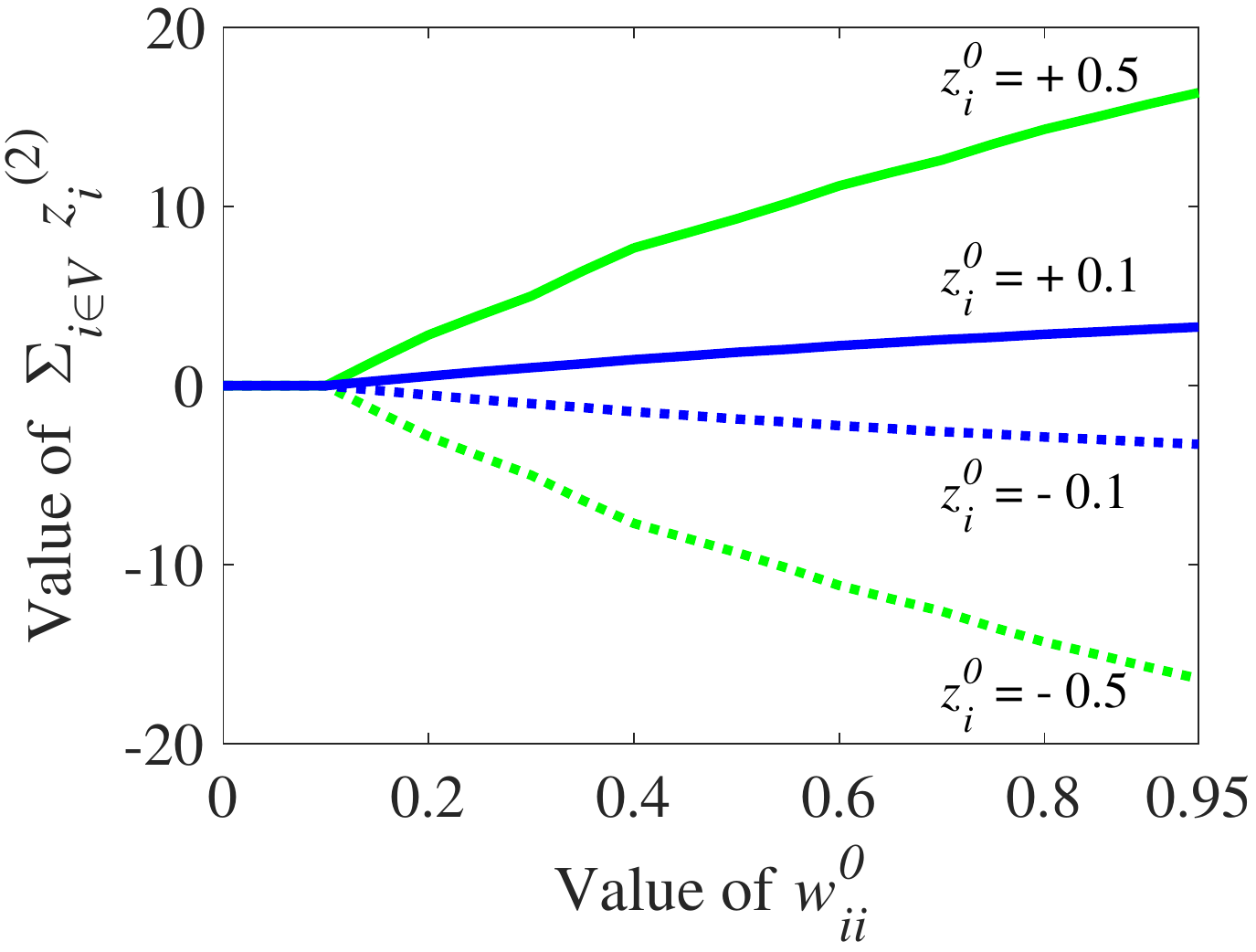}
\vspace{-3mm}
\caption{
The effect of $w_{ii}^0$ 
for different values of $z_{i}^0$ (Karate) with $k_g=k_b=5$ }
\label{fig:different_V0} 
\end{figure}

\begin{figure}[h!]
\vspace{-2mm}
\centering
\includegraphics[width=0.6\textwidth]{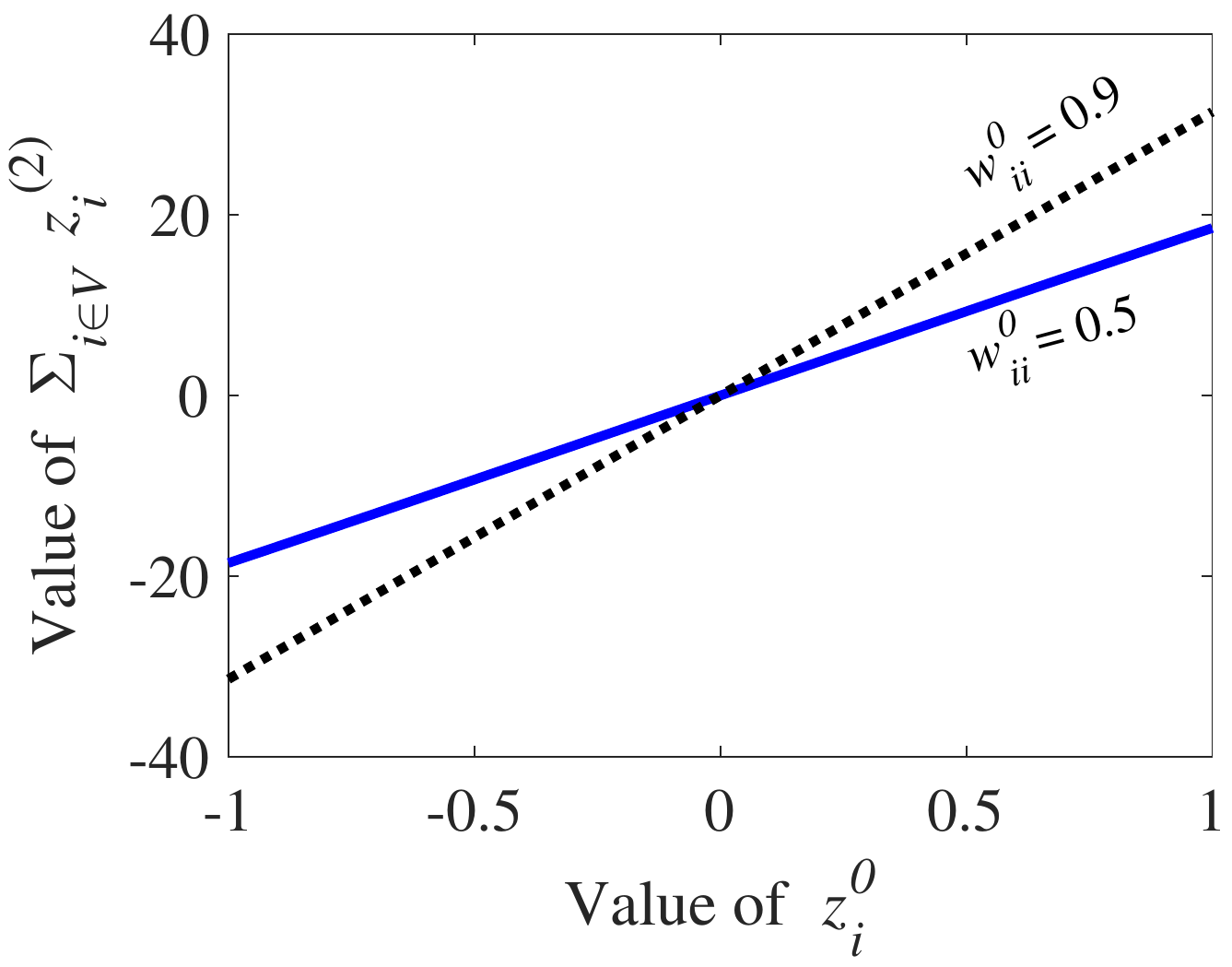}
\vspace{-3mm}
\caption{
The effect of $z_{i}^0$ 
for different values of $w_{ii}^0$ (Karate) with $k_g=k_b=5$ }
\label{fig:with_V0}
\end{figure}

Figures \ref{fig:different_V0} and \ref{fig:with_V0} present the effect of $z_i^0$ in conjunction with $w_{ii}^0$, on  the value of $\sum_{i\in V} z_i^{(2)}$.
We  observed that a higher value of initial bias in the first phase ($z_i^0$) results in a better utility for the good camp (and hence worse utility for the bad camp). This is not only because of a head start offered to the good camp, but also because of its investment being more effective.
This is further magnified for higher $w_{ii}^0$, that is, if nodes attribute higher weightage to their biases (similar to Section \ref{sec:phasewise} where a higher $w_{ii}^0$  results in a more significant change in opinion values).

It is also noteworthy to see that, if we change the sign of $z_i^0$ without changing its magnitude, the value of $\sum_{i\in V} z_i^{(2)}$ also changes its sign while maintaining its magnitude.
This is owing to the fact that, changing the sign of $z_i^0$ interchanges the two camps' roles, that is, the values of $w_{ig}^{(1)}$ and $w_{ib}^{(1)}$ (the effectiveness of their investments in the first phase) get interchanged 
(while other parameters such as $w_{ij}$ and $w_{ii}^0$ are common to both camps).

\subsubsection{The effect of a camp deviating from  Nash equilibrium strategy}

We study  the loss incurred by a camp if it deviates from its Nash equilibrium strategy, to:
(a) a myopic strategy (investing by perceiving its utility as $\sum_{i\in V} z_i^{(1)}$), or (b) the farsighted single camp strategy (investing by perceiving its utility  as $\sum_{i\in V} z_i^{(2)}$ but ignoring the presence of the competing  camp).
In our simulations, we consider  the good camp deviating from its equilibrium strategy (while the bad camp played its equilibrium strategy).

\begin{figure}[h!]
\centering
\includegraphics[width=0.6\textwidth]{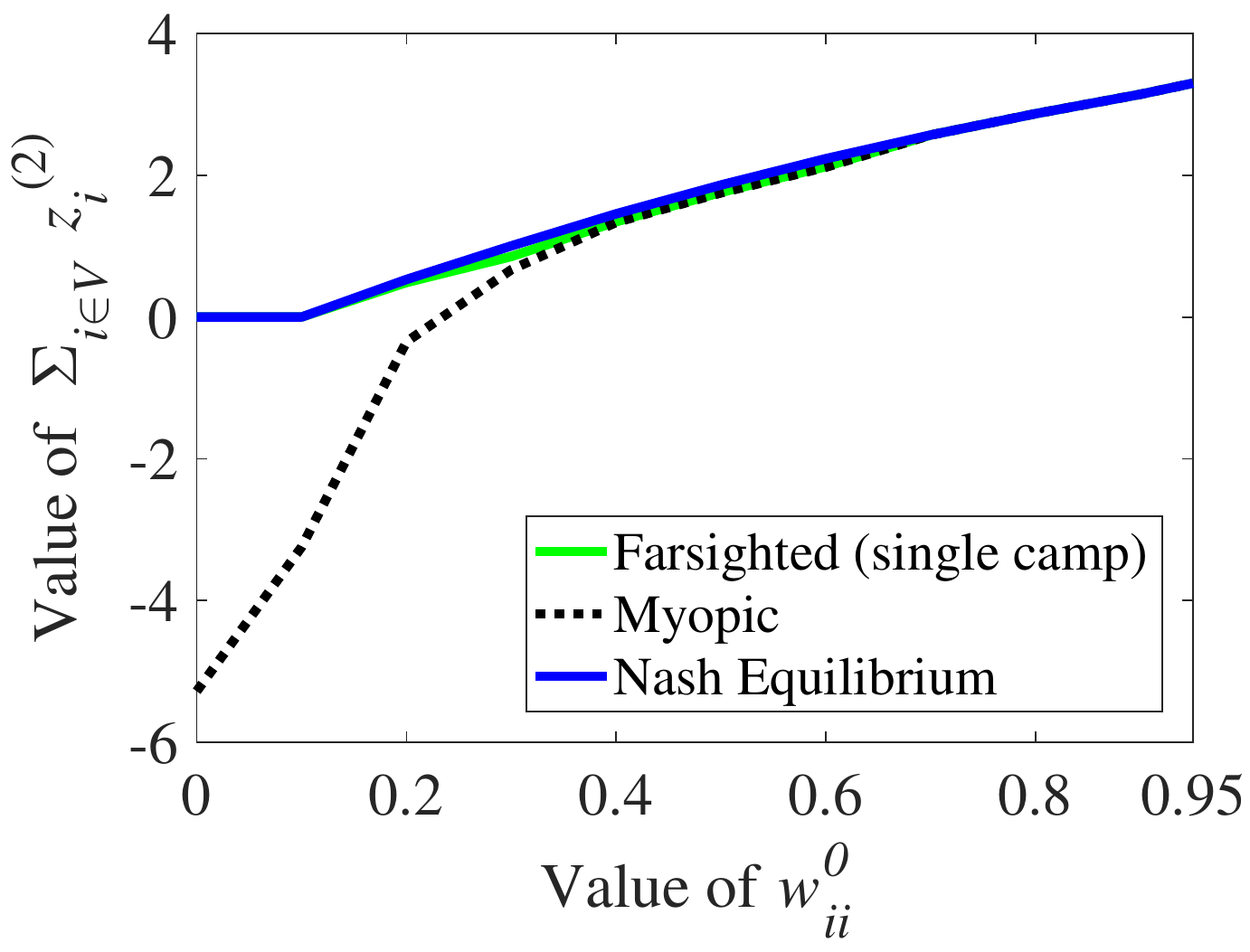}
\vspace{-3mm}
\caption{
The effect of the good camp deviating from its Nash equilibrium strategy
 for Karate dataset with $k_g=k_b=5$ ($z_i^0 = +0.1,\forall i\in V$)}
\label{fig:MSNE_deviate}
\end{figure}

Note that if the good camp decides to play myopic strategy, it perceives its utility as per Equation (\ref{eqn:sumatP_1}), where its optimal strategy  is independent of the strategy played by the bad camp.
In Figure \ref{fig:MSNE_deviate},
we observe that the loss incurred by playing myopic strategy is significant for low values of $w_{ii}^0$, in which range, it is actually optimal to invest most of the budget in the second phase. 
There is no loss incurred for high values of $w_{ii}^0$, since in this range, its equilibrium strategy is to invest the entire budget in the first phase (Figure \ref{fig:MSNE_dependency}), which also coincides with the myopic strategy (the nodes selected for investing on in both strategies are the same for high $w_{ii}^0$).

The loss incurred by playing the single camp farsighted strategy is observed to be relatively insignificant (the values of $\sum_{i\in V} z_i^{(2)}$ corresponding to it almost coincide  with the equilibrium values in Figure~\ref{fig:MSNE_deviate}). However, it is likely to be a property of the small-sized Karate club network, rather than a generalizable observation.
It would be interesting to conduct a study on larger networks, for which
more efficient algorithms need to be developed to find
 Nash equilibrium.

\vspace{-1mm}
\section{Conclusion
}

Using Friedkin-Johnsen model of opinion dynamics, we proposed a framework for two-phase investment 
on nodes in a social network,
where a node's final opinion in the first phase acts as its   bias in the second phase,
and the effectiveness of a camp's investment on that node in the second phase
depends on this bias. 
We formulated a two-phase investment game, where the camps' utilities involved a  parameter which can be interpreted as two-phase Katz centrality.

For the case when there is  one investing camp, we derived polynomial time algorithm for determining an optimal way to split the budget  between the two  phases. 
We observed a natural tradeoff, since
a lower investment in the first phase results in worse initial biases for the second phase, while a higher investment in the first phase spares a lower available budget for the second phase.
Our simulations quantified the impact of the weightage  that nodes attribute to their  biases.
A high weightage  necessitated high investment in the first phase,
so as to effectively influence the biases to be harnessed in the second phase.
We also showed the loss incurred by a camp when it uses a myopic strategy instead of the derived optimal one, thus highlighting the importance of an optimal budget split.
We also illustrated the phasewise progression of opinion values, and hence observed the significant change in opinion values in the second phase when nodes attribute high weightage to their biases.
In order to circumvent iterating over all pairs of nodes for finding an optimal solution, we studied other possible approaches and compared their performance with our optimal approach.

For the case when there are two competing camps, we showed  existence of Nash equilibrium  and
its polynomial time computability under reasonable  assumptions ($w_{ij}  \geq  0, \forall (i,j)$ and $w_{ii}^0  \geq  0,\theta_i  \geq  0, z_i^0  \in  [-1, 1], \forall {i \in V}$).
Using simulations, we computed the mixed strategy Nash equilibrium investments of the camps, and hence observed the  expectation of the first phase investments with respect to the weightage attributed by nodes to their biases.
A similar trend was seen as the single camp case, that a high weightage  necessitates high investment in the first phase.
We also observed that a higher value of initial bias in the first phase results in a better utility for the good camp (and hence worse utility for the bad camp), not only because of a head start but also because of its investment being more effective.
This is further magnified if nodes attribute higher weightage to their biases.
We concluded by showing the loss incurred by a camp if it deviated from its Nash equilibrium strategy.

\subsection*{Future Work}

This work has several interesting directions for future work, of which we mention a few.
It would be interesting to study the problem with bounds on investment on each node by the two camps (such as $\forall i \in V,\, x_i+y_i \leq 1$).
It is worth exploring whether there exist more efficient algorithms which can circumvent iterating over all pairs of nodes, while giving an approximation guarantee to the optimal solution and Nash equilibrium derived in this paper.
This would enable a more extensive study on larger datasets and see if any possible additional insights could be obtained.
%
%
For the competitive case, it would also be interesting to analyze solution concepts other than Nash equilibrium, such as correlated equilibrium.
 
The two-phase study can be generalized to multiple phases to see if any additional insights or benefits can be obtained.
We studied such a problem in our another work \cite{dhamal2018optmulti}, in which we assumed a camp's influence on a node to be independent of the node's bias.
However, with this assumption relaxed in our current work, the analysis for more than two phases complicates considerably. This work could act as the foundation for  a generalized study.

This paper studied a setting where the two camps were represented by positive and negative opinions values. However, in general, we would need to have the camps hold opinions in a multidimensional plane rather than on the real number line (e.g., \cite{dhamal2019integrated}), especially when there are more than two camps.
We explain one way in which this could be done.
Let each camp (say $h$)  have a vector associated with its opinion (say $\vec{v}_h$). Let its investment on node $i$ be denoted by $q_{hi}$ and the weightage attributed by  $i$ to the camp's opinion be $w_{ih}$.
One would need to derive the expression for  the vector-sum of nodes' opinions at the end of the second phase, $\sum_{i\in V} \vec{z}_i^{(2)}$.
 A camp's objective would be
to drive this vector-sum towards the direction of its own opinion vector, that is,
to maximize the inner product between the vector-sum of nodes' opinions and its own opinion vector (that is, $\vec{v}_h \cdot \sum_{i \in V} \vec{z}_i^{(2)}$).
Our  paper can be seen as studying the special case where we have two camps: $h=g,b$ with $\vec{v}_g = +1, \vec{v}_b = -1$ and $q_{gi}=x_i, q_{bi}=y_i$.

This paper considered Friedkin-Johnsen model of opinion dynamics, since it accounts for nodes' biases which are fundamental to the studied problem (as they are responsible for the change in weightage attributed to the camps).
So, in principle, the studied problem can be extended to any model which accounts for nodes' biases in the process of opinion dynamics.
While the linearity of  Friedkin-Johnsen model enabled us to arrive at closed form expressions and analytically derive the solutions, it would be worth investigating if the analysis is tractable when considering other such models.
This work studied a setting where the camps could tune their investments at the end of the first phase, that is, after the opinion dynamics reaches convergence.
However, the two-phase setting can be extended to a dynamical system by allowing camps to tune their investments in each iteration of the opinion dynamics process, and the total amount invested depends on the time for which the investment is made.

Personalized PageRank \cite{haveliwala2002topic,jeh2003scaling} is likely to be of interest while relating it to external camp weights, since these weight vectors ($\mathbf{w_g}$ and $\mathbf{w_b}$) are personalized according to the camps. 
Also, Dynamic PageRank \cite{gleich2014dynamical} could be related to the changing external camp weights over different phases.
It is worth exploring if efficient and effective heuristics can be developed along these lines.
%
%
In our study, we considered that the weightage attributed by a node to a camp changes depending on how aligned the node's bias is towards the camp. One could also consider that the weightage attributed by a node to its neighbor changes depending on how aligned its bias is towards its neighbor's opinion. 
We did not consider this since the weights between nodes stabilize after developing over months and years of interactions, while the camps are relatively strangers to the nodes and the weightage attributed to them are seldom stable.
However, in case of such a study which considers graphs and weights that change over time, it would be useful to define dynamic two-phase Katz centrality (on similar lines as we defined two-phase Katz centrality and taking cue from dynamic Katz centrality \cite{grindrod2011communicability}) for determining the importance of the individual nodes in opinion dynamics in the different phases.

\section*{Acknowledgment}

This work is partly supported by CEFIPRA grant No. IFC/DST-Inria-2016-
01/448 ``Machine Learning for Network Analytics''.

\section*{References}

\begin{small}
\bibliography{ODSN_Multiphase_references}

\begin{thebibliography}{10}
\expandafter\ifx\csname url\endcsname\relax
  \def\url#1{\texttt{#1}}\fi
\expandafter\ifx\csname urlprefix\endcsname\relax\def\urlprefix{URL }\fi
\expandafter\ifx\csname href\endcsname\relax
  \def\href#1#2{#2} \def\path#1{#1}\fi

\bibitem{dhamal2018manipulating}
S.~Dhamal, W.~Ben-Ameur, T.~Chahed, E.~Altman, Manipulating opinion dynamics in
  social networks in two phases, in: The Joint International Workshop on Social
  Influence Analysis and Mining Actionable Insights from Social Networks, 2018.

\bibitem{easley2010networks}
D.~Easley, J.~Kleinberg, Networks, crowds, and markets, Cambridge University
  Press, 2010.

\bibitem{acemoglu2011opinion}
D.~Acemoglu, A.~Ozdaglar, Opinion dynamics and learning in social networks,
  Dynamic Games and Applications 1~(1) (2011) 3--49.

\bibitem{gionis2013opinion}
A.~Gionis, E.~Terzi, P.~Tsaparas, Opinion maximization in social networks, in:
  2013 International Conference on Data Mining, SIAM, 2013, pp. 387--395.

\bibitem{grabisch2017strategic}
M.~Grabisch, A.~Mandel, A.~Rusinowska, E.~Tanimura, Strategic influence in
  social networks, Mathematics of Operations Research 43~(1) (2018) 29--50.

\bibitem{dhamal2018framework}
S.~Dhamal, W.~Ben-Ameur, T.~Chahed, E.~Altman, Optimal investment strategies
  for competing camps in a social network: A broad framework, IEEE Transactions
  on Network Science and Engineering (2018) in press.

\bibitem{friedkin1990social}
N.~Friedkin, E.~Johnsen, Social influence and opinions, Journal of Mathematical
  Sociology 15~(3-4) (1990) 193--206.

\bibitem{friedkin1997social}
N.~Friedkin, E.~Johnsen, Social positions in influence networks, Social
  Networks 19~(3) (1997) 209--222.

\bibitem{krause2000discrete}
U.~Krause, A discrete nonlinear and non-autonomous model of consensus
  formation, Communications in Difference Equations (2000) 227--236.

\bibitem{xia2013opinion}
H.~Xia, H.~Wang, Z.~Xuan, Opinion dynamics: A multidisciplinary review and
  perspective on future research, in: Multidisciplinary Studies in Knowledge
  and Systems Science, IGI Global, 2013, pp. 311--332.

\bibitem{kempe2003maximizing}
D.~Kempe, J.~Kleinberg, {\'E}.~Tardos, Maximizing the spread of influence
  through a social network, in: 9th International Conference on Knowledge
  Discovery and Data Mining, ACM, 2003, pp. 137--146.

\bibitem{guille2013information}
A.~Guille, H.~Hacid, C.~Favre, D.~A. Zighed, Information diffusion in online
  social networks: A survey, ACM SIGMOD Record 42~(1) (2013) 17--28.

\bibitem{gomez2016influence}
M.~Gomez-Rodriguez, L.~Song, N.~Du, H.~Zha, B.~Sch{\"o}lkopf, Influence
  estimation and maximization in continuous-time diffusion networks, ACM
  Transactions on Information Systems (TOIS) 34~(2) (2016) 9:1--9:33.

\bibitem{farajtabar2017coevolve}
M.~Farajtabar, Y.~Wang, M.~Gomez-Rodriguez, S.~Li, H.~Zha, {COEVOLVE}: A joint
  point process model for information diffusion and network evolution, Journal
  of Machine Learning Research 18 (2017) 1--49.

\bibitem{lorenz2007continuous}
J.~Lorenz, Continuous opinion dynamics under bounded confidence: A survey,
  International Journal of Modern Physics C 18~(12) (2007) 1819--1838.

\bibitem{gomez2015diffusion}
M.~Gomez~Rodriguez, L.~Song, Diffusion in social and information networks:
  Research problems, probabilistic models and machine learning methods, in:
  Proceedings of the 21th ACM SIGKDD International Conference on Knowledge
  Discovery and Data Mining, ACM, 2015, pp. 2315--2316.

\bibitem{degroot1974reaching}
M.~H. DeGroot, Reaching a consensus, Journal of the American Statistical
  Association 69~(345) (1974) 118--121.

\bibitem{holley1975ergodic}
R.~Holley, T.~Liggett, Ergodic theorems for weakly interacting infinite systems
  and the voter model, The Annals of Probability (1975) 643--663.

\bibitem{yildiz2013binary}
E.~Yildiz, A.~Ozdaglar, D.~Acemoglu, A.~Saberi, A.~Scaglione, Binary opinion
  dynamics with stubborn agents, ACM Transactions on Economics and Computation
  1~(4) (2013) 19:1--19:30.

\bibitem{lynn2016maximizing}
C.~Lynn, D.~D. Lee, Maximizing influence in an {I}sing network: A mean-field
  optimal solution, in: Advances in Neural Information Processing Systems,
  2016, pp. 2495--2503.

\bibitem{rossi2015role}
R.~A. Rossi, N.~K. Ahmed, Role discovery in networks, IEEE Transactions on
  Knowledge and Data Engineering 27~(4) (2015) 1112--1131.

\bibitem{abiteboul2003adaptive}
S.~Abiteboul, M.~Preda, G.~Cobena, Adaptive on-line page importance
  computation, in: Proceedings of the 12th international conference on World
  Wide Web, ACM, 2003, pp. 280--290.

\bibitem{grindrod2011communicability}
P.~Grindrod, M.~C. Parsons, D.~J. Higham, E.~Estrada, Communicability across
  evolving networks, Physical Review E 83~(4) (2011) 046120.

\bibitem{gleich2014dynamical}
D.~F. Gleich, R.~A. Rossi, A dynamical system for pagerank with time-dependent
  teleportation, Internet Mathematics 10~(1-2) (2014) 188--217.

\bibitem{myers2012clash}
S.~A. Myers, J.~Leskovec, Clash of the contagions: Cooperation and competition
  in information diffusion, in: Proceedings of the Twelfth International
  Conference on Data Mining (ICDM), IEEE, 2012, pp. 539--548.

\bibitem{tzoumas2012game}
V.~Tzoumas, C.~Amanatidis, E.~Markakis, A game-theoretic analysis of a
  competitive diffusion process over social networks, in: International
  Workshop on Internet and Network Economics, Springer, 2012, pp. 1--14.

\bibitem{etesami2016complexity}
S.~R. Etesami, T.~Ba{\c{s}}ar, Complexity of equilibrium in competitive
  diffusion games on social networks, Automatica 68 (2016) 100--110.

\bibitem{bharathi2007competitive}
S.~Bharathi, D.~Kempe, M.~Salek, Competitive influence maximization in social
  networks, in: 3rd International Workshop on Web and Internet Economics,
  Springer, 2007, pp. 306--311.

\bibitem{goyal2014competitive}
S.~Goyal, H.~Heidari, M.~Kearns, Competitive contagion in networks, Games and
  Economic Behavior.

\bibitem{anagnostopoulos2015competitive}
A.~Anagnostopoulos, D.~Ferraioli, S.~Leonardi, Competitive influence in social
  networks: Convergence, submodularity, and competition effects, in: 14th
  International Conference on Autonomous Agents \& Multiagent Systems, IFAAMAS,
  2015, pp. 1767--1768.

\bibitem{ghaderi2014opinion}
J.~Ghaderi, R.~Srikant, Opinion dynamics in social networks with stubborn
  agents: Equilibrium and convergence rate, Automatica 50~(12) (2014)
  3209--3215.

\bibitem{dubey2006competing}
P.~Dubey, R.~Garg, B.~De~Meyer, Competing for customers in a social network:
  The quasi-linear case, in: 2nd International Workshop on Web and Internet
  Economics, Springer, 2006, pp. 162--173.

\bibitem{bimpikis2016competitive}
K.~Bimpikis, A.~Ozdaglar, E.~Yildiz, Competitive targeted advertising over
  networks, Operations Research 64~(3) (2016) 705--720.

\bibitem{singer2016influence}
Y.~Singer, Influence maximization through adaptive seeding, ACM SIGecom
  Exchanges 15~(1) (2016) 32--59.

\bibitem{golovin2011adaptive}
D.~Golovin, A.~Krause, Adaptive submodularity: theory and applications in
  active learning and stochastic optimization, Journal of Artificial
  Intelligence Research 42~(1) (2011) 427--486.

\bibitem{seeman2013adaptive}
L.~Seeman, Y.~Singer, Adaptive seeding in social networks, in: 54th Annual IEEE
  Symposium on Foundations of Computer Science, IEEE, 2013, pp. 459--468.

\bibitem{rubinstein2015approximability}
A.~Rubinstein, L.~Seeman, Y.~Singer, Approximability of adaptive seeding under
  knapsack constraints, in: 16th ACM Conference on Economics and Computation,
  ACM, 2015, pp. 797--814.

\bibitem{horel2015scalable}
T.~Horel, Y.~Singer, Scalable methods for adaptively seeding a social network,
  in: 24th International Conference on World Wide Web, ACM, 2015, pp. 441--451.

\bibitem{correa2015adaptive}
J.~Correa, M.~Kiwi, N.~Olver, A.~Vera, Adaptive rumor spreading, in: 11th
  International Conference on Web and Internet Economics, Springer, 2015, pp.
  272--285.

\bibitem{badanidiyuru2016locally}
A.~Badanidiyuru, C.~Papadimitriou, A.~Rubinstein, L.~Seeman, Y.~Singer, Locally
  adaptive optimization: Adaptive seeding for monotone submodular functions,
  in: 27th ACM-SIAM Symposium on Discrete Algorithms, SIAM, 2016, pp. 414--429.

\bibitem{dhamal2016information}
S.~Dhamal, K.~J. Prabuchandran, Y.~Narahari, Information diffusion in social
  networks in two phases, IEEE Transactions on Network Science and Engineering
  3~(4) (2016) 197--210.

\bibitem{dhamal2018effectiveness}
S.~Dhamal, Effectiveness of diffusing information through a social network in
  multiple phases, in: IEEE Global Communications Conference, IEEE, 2018, pp.
  1--7.

\bibitem{tong2016adaptive}
G.~Tong, W.~Wu, S.~Tang, D.-Z. Du, Adaptive influence maximization in dynamic
  social networks, IEEE/ACM Transactions on Networking 25~(1) (2016) 112--125.

\bibitem{tang2016no}
J.~Yuan, S.~Tang, No time to observe: Adaptive influence maximization with
  partial feedback, in: 26th International Joint Conference on Artificial
  Intelligence, 2017, pp. 3908--3914.

\bibitem{sun2018multi}
L.~Sun, W.~Huang, P.~S. Yu, W.~Chen, Multi-round influence maximization, in:
  Proceedings of the 24th ACM SIGKDD International Conference on Knowledge
  Discovery \& Data Mining, ACM, 2018, pp. 2249--2258.

\bibitem{mondal2017two}
S.~Mondal, S.~Dhamal, Y.~Narahari, Two-phase influence maximization in social
  networks with seed nodes and referral incentives, in: International AAAI
  Conference on Web and Social Media, AAAI, 2017, pp. 620--623.

\bibitem{dhamal2018optmulti}
S.~Dhamal, W.~Ben-Ameur, T.~Chahed, E.~Altman, Optimal multiphase investment
  strategies for influencing opinions in a social network, in: 17th
  International Conference on Autonomous Agents \& Multiagent Systems, IFAAMAS,
  2018, pp. 1927--1929.

\bibitem{altafini2013consensus}
C.~Altafini, Consensus problems on networks with antagonistic interactions,
  IEEE Transactions on Automatic Control 58~(4) (2013) 935--946.

\bibitem{proskurnikov2016opinion}
A.~V. Proskurnikov, A.~S. Matveev, M.~Cao, Opinion dynamics in social networks
  with hostile camps: Consensus vs. polarization, IEEE Transactions on
  Automatic Control 61~(6) (2016) 1524--1536.

\bibitem{katz1953new}
L.~Katz, A new status index derived from sociometric analysis, Psychometrika
  18~(1) (1953) 39--43.

\bibitem{boyd2004convex}
S.~Boyd, L.~Vandenberghe, Convex optimization, Cambridge University Press,
  2004.

\bibitem{arrow1958studies}
K.~J. Arrow, L.~Hurwicz, H.~Uzawa, H.~B. Chenery, Studies in linear and
  non-linear programming, Stanford University Press, 1958.

\bibitem{osborne2004introduction}
M.~Osborne, An introduction to game theory, Vol.~3, Oxford University Press,
  New York, 2004.

\bibitem{chen2009efficient}
W.~Chen, Y.~Wang, S.~Yang, Efficient influence maximization in social networks,
  in: 15th International Conference on Knowledge Discovery and Data Mining,
  ACM, 2009, pp. 199--208.

\bibitem{chen2010scalable}
W.~Chen, C.~Wang, Y.~Wang, Scalable influence maximization for prevalent viral
  marketing in large-scale social networks, in: 16th International Conference
  on Knowledge Discovery and Data Mining, ACM, 2010, pp. 1029--1038.

\bibitem{zachary1977information}
W.~W. Zachary, An information flow model for conflict and fission in small
  groups, Journal of Anthropological Research 33~(4) (1977) 452--473.

\bibitem{dhamal2019integrated}
S.~Dhamal, An integrated framework for competitive multi-channel marketing of
  multi-featured products, in: Proceedings of the 11th International Conference
  on Communication Systems \& Networks, 2019, pp. 391--394.

\bibitem{haveliwala2002topic}
T.~H. Haveliwala, Topic-sensitive pagerank, in: Proceedings of the 11th
  international conference on World Wide Web, ACM, 2002, pp. 517--526.

\bibitem{jeh2003scaling}
G.~Jeh, J.~Widom, Scaling personalized web search, in: Proceedings of the 12th
  international conference on World Wide Web, ACM, 2003, pp. 271--279.

\end{thebibliography}
\end{small}

\end{document}